\documentclass[journal,12pt,onecolumn]{IEEEtran}
\usepackage[letterpaper, top=1in, bottom=1in, left=1in, right=1in]{geometry}

\usepackage{setspace}

\usepackage{footnote}
\usepackage[table]{xcolor}
\usepackage{times}
\usepackage{epsfig}
\usepackage{amsmath}
\usepackage{amsthm}
\usepackage{amsfonts}
\usepackage{graphicx}
\usepackage{amssymb}
\usepackage{amstext}
\usepackage{latexsym}
\usepackage{color,colortbl}
\usepackage{ifthen}
\usepackage{multirow}
\usepackage{verbatim}
\usepackage{array,tabularx}
\usepackage{arydshln}
\usepackage[mathscr]{euscript}
\usepackage{accents}
 \usepackage{cite}
\usepackage{hhline}
\usepackage{caption}
\usepackage{subcaption}
\usepackage{enumerate}
\usepackage{xcolor}
\usepackage{mathtools}
\usepackage{url}
\usepackage{xparse}
\usepackage{makecell}
\usepackage{varwidth}
\usepackage{bm}
\usepackage{arydshln}
\usepackage{comment}
\usepackage{wrapfig}

\usepackage{caption}
\usepackage{float}
\usepackage{booktabs}

\usepackage{mathtools}

\usepackage{booktabs}

\usepackage{float}

\usepackage{multirow}

\usepackage{algorithm,algpseudocode}

\newcolumntype{C}[1]{>{\centering\let\newline\\\arraybackslash\hspace{0pt}}m{#1}}

\usepackage[normalem]{ulem}

\usepackage{amsmath,pgfplots,amssymb}
\usepackage{subcaption,graphicx}

\newtheorem{theorem}{Theorem}

\newtheorem{lemma}{Lemma}

\theoremstyle{definition}

\newtheorem{definition}{Definition}

\theoremstyle{definition}
\newtheorem{remark}{Remark}

\theoremstyle{definition}

\DeclareMathOperator{\lcm}{lcm}

\usetikzlibrary{matrix,decorations.pathreplacing}

\newcommand{\red}[1] {\textcolor{red}{ #1}}

\definecolor{DarkGreen}{rgb}{0.1,0.5,0.1}
\definecolor{DarkRed}{rgb}{0.5,0.1,0.1}
\definecolor{DarkBlue}{rgb}{0.1,0.1,0.5}
\definecolor{DarkPurple}{rgb}{0.5,0.2,0.5}
\definecolor{DarkTurquoise}{rgb}{0.1,0.5,0.5}

\newcommand{\raf}[1]{\textcolor{DarkGreen}{ [ #1 --rafael ] \normalsize }}

\newcommand{\ale}[1]{\textcolor{DarkRed}{ [ #1 --alejandro ] \normalsize }}

\newcommand{\off}[1]{}

\algnewcommand\algorithmicforeach{\textbf{for each}}
\algdef{S}[FOR]{ForEach}[1]{\algorithmicforeach\ #1\ \algorithmicdo}

\newcommand{\kb}{k_b}			 		%
\newcommand{\ku}{k_u}			 		%
\newcommand{\ks}{k_s}			 		%
\newcommand{\nb}{n_b}			 		%

\ifCLASSINFOpdf
\else
\fi
\begin{document}
%
\title{Network Coding-Based\\ Post-Quantum Cryptography}
%
%
%

\author{\hspace{-0.2cm}Alejandro Cohen, Rafael G. L. D’Oliveira, Salman Salamatian, and Muriel M\'{e}dard \\ Research Laboratory of Electronics, MIT, Cambridge, MA, USA,\\ Emails: \{cohenale, rafaeld, salmansa, medard\}@mit.edu\thanks{Patent application submitted: no. 63/072,430.}}

\maketitle


\begin{abstract}
We propose a novel hybrid universal network-coding cryptosystem (HUNCC) to obtain secure post-quantum cryptography at high communication rates. The secure network-coding scheme we offer is hybrid in the sense that it combines information-theory security with public-key cryptography. In addition, the scheme is general and can be applied to any communication network, and to any public-key cryptosystem.
Our hybrid scheme is based on the information theoretic notion of individual secrecy, which traditionally relies on the assumption that an eavesdropper can only observe a subset of the communication links between the trusted parties -- an assumption that is often challenging to enforce.
For this setting, several code constructions have been developed, where the messages are linearly mixed before transmission over each of the paths in a way that guarantees that an adversary which observes only a subset has sufficient uncertainty about each individual message.

Instead, in this paper, we take a computational viewpoint, and construct a coding scheme in which an arbitrary secure cryptosystem is utilized on a subset of the links, while a pre-processing similar to the one in individual security is utilized.
Under this scheme, we demonstrate 1) a computational security guarantee for an adversary which observes the entirety of the links 2) an information theoretic security guarantee for an adversary which observes a subset of the links, and 3) information rates which approach the capacity of the network and greatly improve upon the current solutions.

A perhaps surprising consequence of our scheme is that, to guarantee a computational security level $b$, it is sufficient to encrypt a single link using a computational post-quantum scheme.
That is, using HUNCC, we can ensure post-quantum security in networks where it is not possible to use public-key encryption over all the links in the network. In addition, the information rate approaches 1 as the number of communication links increases.
As a concrete example, in a multipath network with three links, using a 128-bit computationally secure McEliece cryptosystem only over one link, we obtain a 128-bit computational security level over all paths with a total information rate of 0.91 in the network.

\end{abstract}

%
\IEEEpeerreviewmaketitle

  \section{Introduction}

\begin{figure}[t]
    \centering
    \includegraphics[width = 0.5\columnwidth]{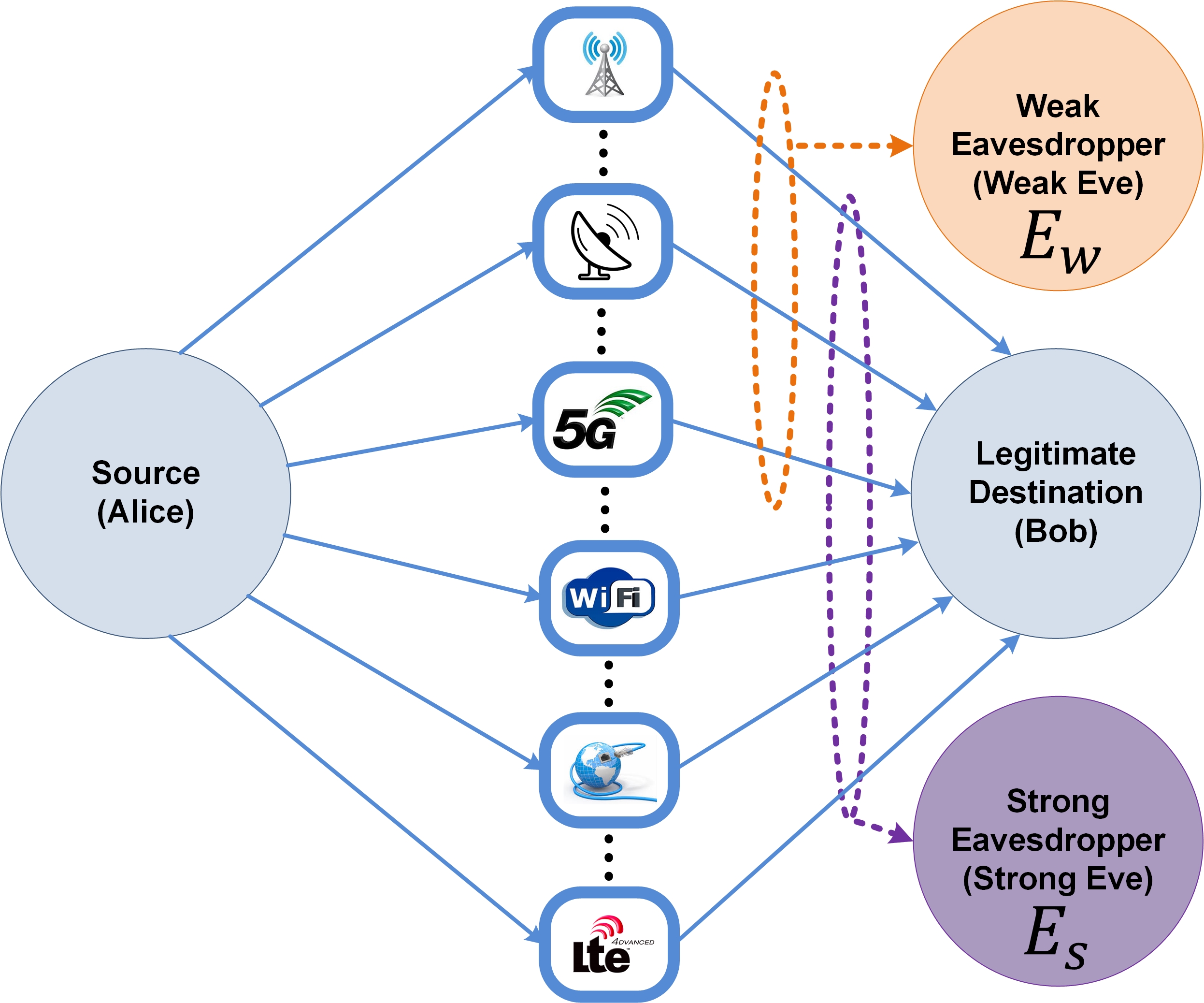}
    \caption{Post-quantum secure multipath network with $l$ paths, one source, Alice, one legitimate destination, Bob, and two types of possible eavesdropper's, Eve, weak and strong, which can obtain the information transmitted over $w<l$ or all the $l$ paths, respectively.}
    \label{fig:system_model}
\end{figure}

The connection between Information Theory and Cryptography goes back to Shannon himself -- it is said that his work in security inspired his seminal work on communications. There, he defines the information-theoretic notion of \emph{perfect secrecy} when studying the setting where two users, Alice and Bob wish to communicate privately in the presence of an eavesdropper, Eve. Under perfect secrecy, $H(M|X)=H(M)$, where $H$ is the entropy, $M$ is the private message, and $X$ is the encrypted message sent through the communication network. If Eve can observe the encrypted message $X$ completely, perfect privacy can only be obtained if both Alice and Bob share a random key, $R$, with entropy as large as the message, i.e. $H(R)\geq H(M)$ \cite{shannon1949communication}.

This necessity in Alice and Bob sharing large secret keys is often non-practical, e.g. if Alice and Bob are geographically distant. Hence, much effort has been devoted to developing alternative solutions which relax upon the perfect secrecy condition. One such relaxation comes from assuming that the eavesdropper has limited computational power. Such schemes, referred to as \emph{computationally secure} in this paper, rely on the conjecture that certain one-way functions are hard to invert \cite{kaltz2008introduction}. Thus, Alice encrypts the private message using a one-way function before sending it to Bob. This function should be hard for Eve to invert, but inversion for Bob should be possible if he possesses the right key. One way to achieve this is via \emph{public-key cryptography}.

A public-key cryptosystem consists of an encryption function $\mathrm{Enc}(\cdot)$, a decryption function $\mathrm{Dec}(\cdot)$, a secret key $s$, and a public key $p$. The encryption function uses the public key to encrypt the private message $M$ into $\mathrm{Enc}(M,p)$. The decryption function uses the secret key to decrypt the encrypted message $M = \mathrm{Dec}[\mathrm{Enc}(M,p),s]$. The critical property is that decrypting the encrypted message without the secret key is computationally expensive. Although there are many ways of characterizing this, we will focus on the notion of \emph{security level}. Informally, a public-key cryptosystem has security level $b$, referred to as being $b$-bit secure, if the amount of operations expected to decode the encrypted message, without knowledge of the secret key, is~$2^b$.

One of the first, and most widely used, public-key cryptosystems is the Rivest–Shamir–Adleman (RSA) cryptosystem \cite{rivest1978method}. The security of RSA relies on the hardness conjecture of two mathematical problems: integer factorization and the RSA problem. In 1994 however, Peter Shor presented a polynomial-time algorithm for integer factorization, known as Shor's algorithm \cite{365700} -- with the caveat that the algorithm runs on a quantum computer. In other words, if sufficiently large quantum computers are ever to be built, Shor's algorithm can be used to break the RSA cryptosystem \cite{hallgren2005fast,hallgren2007polynomial,schmidt2005polynomial,ding2005cryptanalysis,shor1999polynomial}. This subsequently led to an increased interest in cryptosystems which are resilient to quantum attacks, a field known as \emph{post-quantum cryptography} \cite{bernstein2008attacking,bernstein2009introduction}.

A candidate for post-quantum cryptography, known as the \emph{McEliece cryptosystem}, was introduced in \cite{mceliece1978public}. The McEliece cryptosystem is immune to attacks that use Shor's algorithm. Interestingly, the McEliece is also based on a connection between cryptography and communication theory --
its security relies on the fact that decoding a general linear code is NP-hard \cite{1055873}. The original scheme uses binary Goppa codes \cite{berlekamp1973goppa,1055350}. Apart from having no known quantum attacks, the encryption and decryption algorithms are faster than those of RSA. Two main disadvantages of McEliece are: i) in usual applications, the size of the public key is much larger than that of RSA and ii) it suffers from a large communication overhead, with a communication rate around $0.5$ in the original paper.

One may be tempted to increase the communication rate by changing the parameters of the Goppa code.
However, a key result pertaining to the communication rate was presented in \cite{6089437}, where a polynomial time algorithm was given for distinguishing the matrix of a high rate Goppa code with a random matrix. Therefore, the security of high rates Goppa code may not be guaranteed. Another idea is to look at other families of codes, away from Goppa codes. As shown in Table \ref{tab:code}, most of them have been broken.

In parallel with the advances in computational security, another relaxation on perfect privacy was considered in the literature, mostly by information theorists \cite{wyner1975wire,ozarow1985wire,el2012secure,el2007wiretap,liang2009information,bloch2011physical,zhou2013physical}. Instead of restricting the computational power of the eavesdropper, in \emph{physical layer security} one limits how much information Eve can obtain about the encrypted message. In this setting, information-theoretic security can be obtained at the expense of the communication rate, as shown in the seminal work of Wyner\cite{wyner1975wire}, where he introduced the wiretap channel -- the analog of the classical Alice, Bob and Eve triple under the physical layer security assumption. For example, in \cite{6770743} it is assumed that the eavesdropper can observe any set $w$ out of a total of $n$ transmitted symbols. Denoting this set by $Y_{E_w}$, it was shown that there exist encryption codes with communication rate $\frac{n-w}{n}$ which does not leak any information about the message to the eavesdropper, i.e. $H(M|Y_{E_w}) = H(M)$. The price to pay in rate to achieve perfect secrecy in physical layer security is significant -- decreasing the rate of the legitimate communication is necessary.

In an effort to increase the efficiency in terms of rate, yet another relaxation of the perfect secrecy of Shannon was introduced, namely \emph{individual secrecy} \cite{kobayashi2013secure,bhattad2005weakly,silva2009universal,silva2011universal,mansour2014secrecy,chen2015individual,mansour2015individual,mansourindividual, goldenbaum2015multiple,chensecure, mansour2015individual1,cohen2018secure}.
This is best explained in a network setup, where Alice has many messages to send to Bob, say $M_1,\ldots, M_m$, and Eve may observe any $w$ of them. By increasing the rate beyond the limits given by the wiretap channel of Wyner, it is inevitable that information will be leaked.
Yet, it is unclear whether the eavesdropper Eve is able to utilize this information to recover the sent messages.

The notion of individual secrecy crystallizes this concept by guaranteeing that $H(M_i|Y_{E,w}) = H(M_i)$, for all $i = 1, \ldots, m$.
Note that this is generally weaker than perfect secrecy, where $H(M_1,\ldots, M_m |Y_{E,w}) = H(M_1,\ldots, M_m)$.
In other words, the information that Eve obtains from her observation does not help her decipher each individual message -- rather it is information about the combination of the messages.
Individual secrecy is definitely a weaker notion of security, yet it often allows to increase the rate drastically, sometimes even making the encryption process free in terms of rate.
While there is no doubt that efficient rates are beneficial, the assumptions of physical layer security, namely that Eve does not experience a worse channel than Bob, are hard to enforce in practice. Figure \ref{fig:foobar} depicts a simple example of secure post-quantum candidates for multipath networks based on the following security solutions considered in the literature we elaborated above, (a) One-time pad \cite{shannon1949communication} (b) McEliece cryptosystem \cite{mceliece1978public} (c) Network coding wiretap II \cite{el2007wiretap} (d) Individual Security for networks \cite{cohen2018secure}.

\begin{figure}[t]
\begin{flushleft}
\centering
\begin{subfigure}[b]{0.32\textwidth}
\includegraphics[scale=0.4]{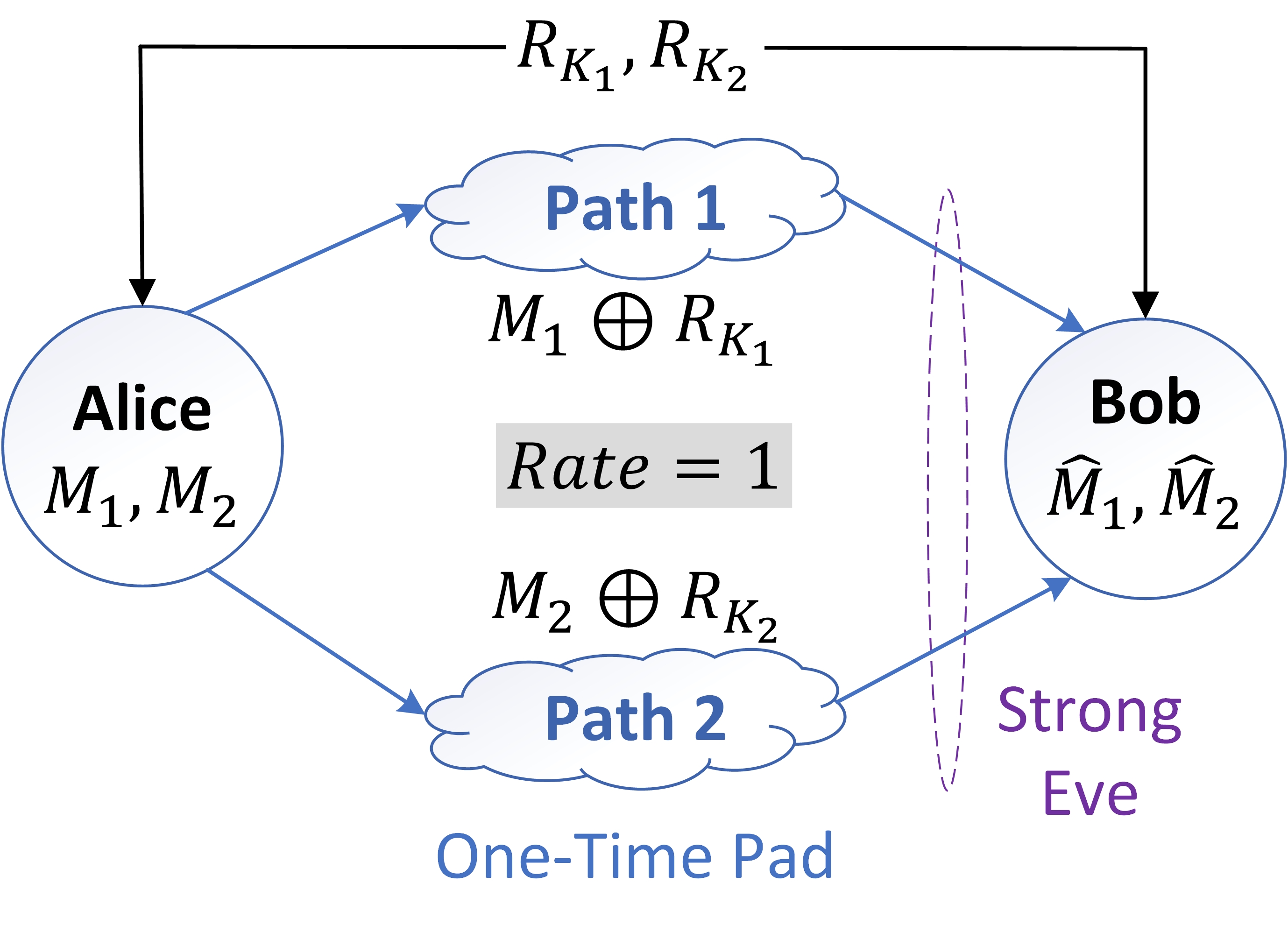}
\caption{}
\end{subfigure}
\begin{subfigure}[b]{0.32\textwidth}
\includegraphics[scale=0.4]{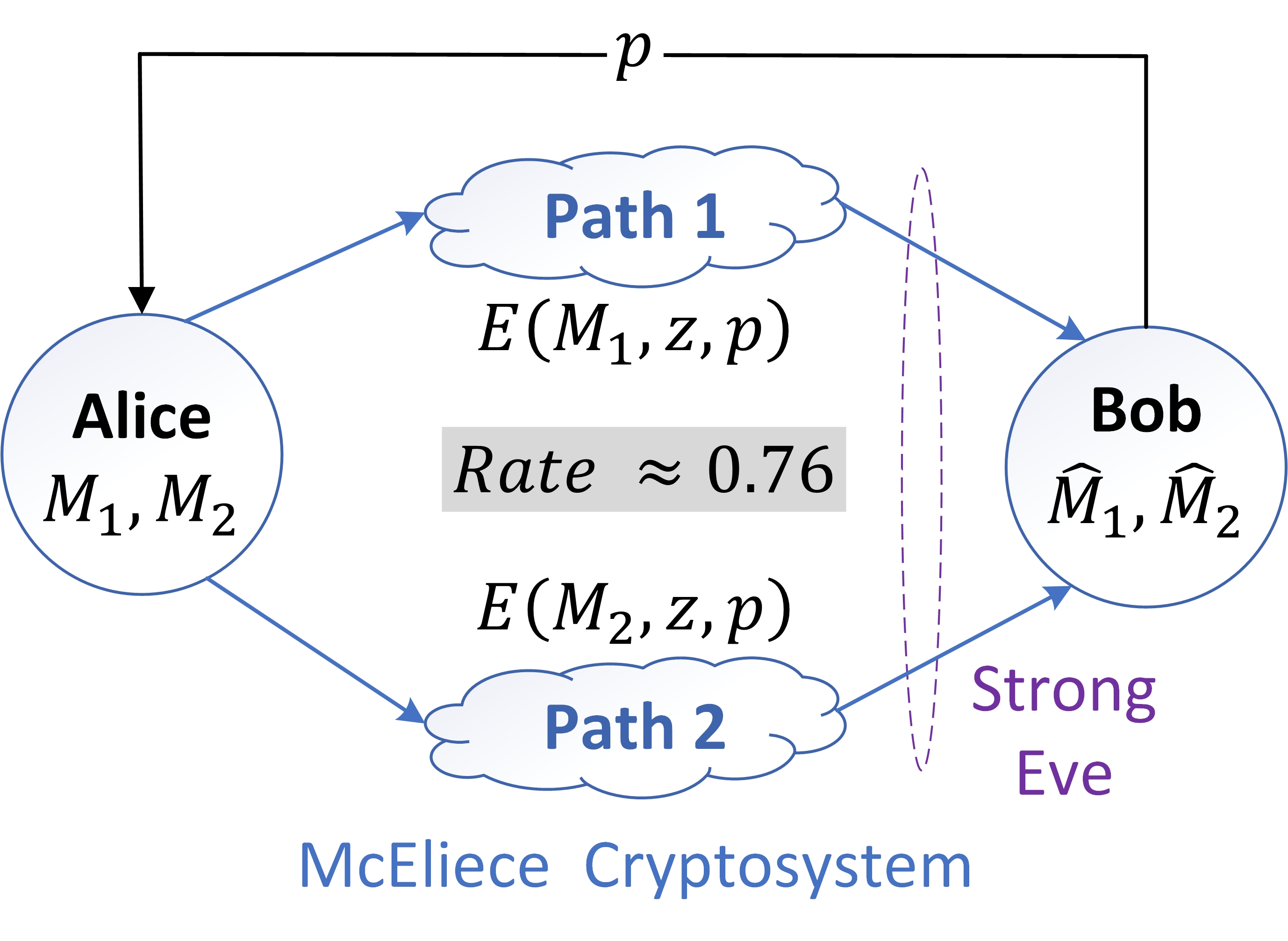}
\caption{}
\end{subfigure}
\begin{subfigure}[b]{0.32\textwidth}
\includegraphics[scale=0.4]{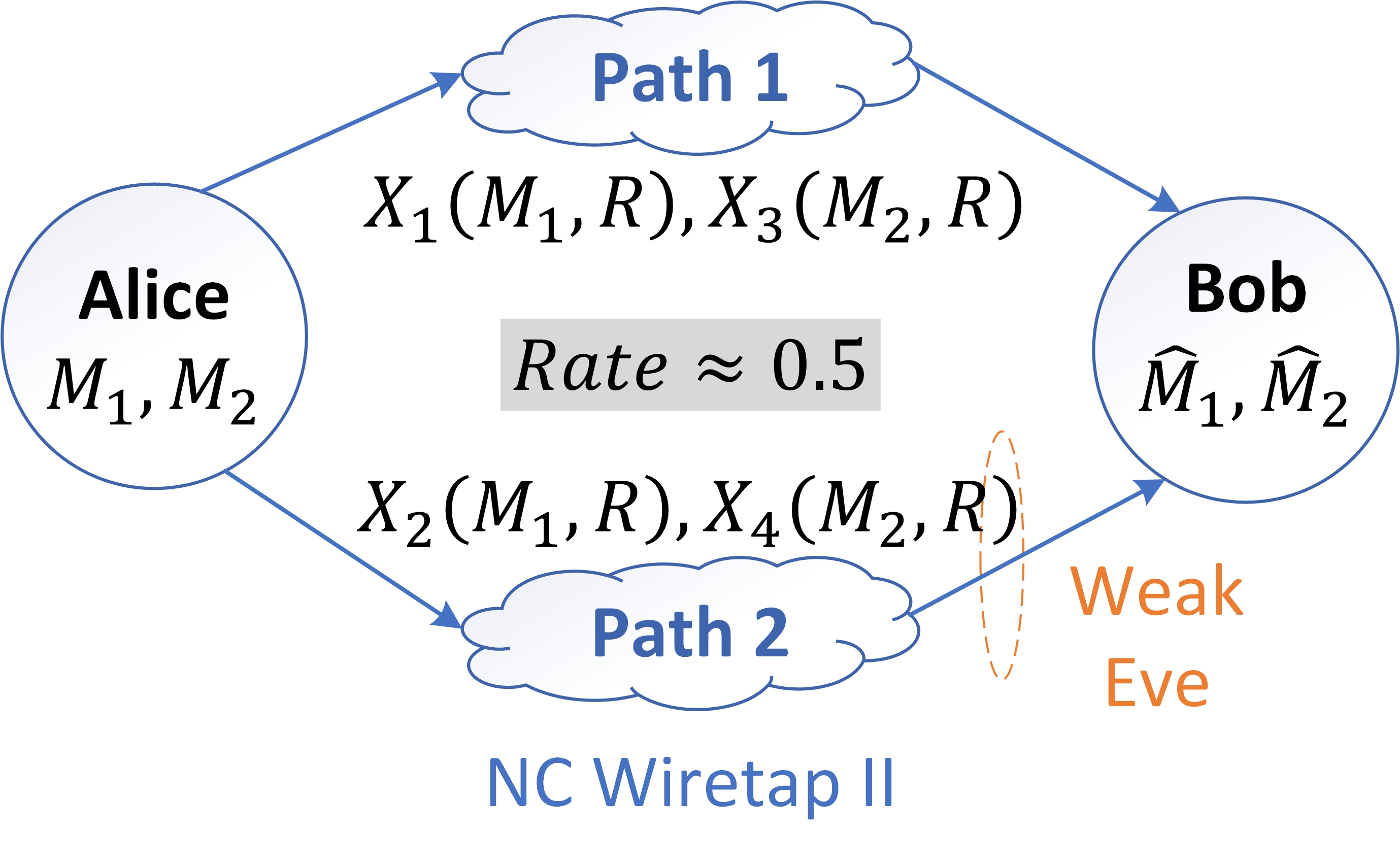}
\caption{}
\end{subfigure}
\end{flushleft}
\begin{flushleft}
\centering
\begin{subfigure}[b]{0.32\textwidth}
\includegraphics[scale=0.4]{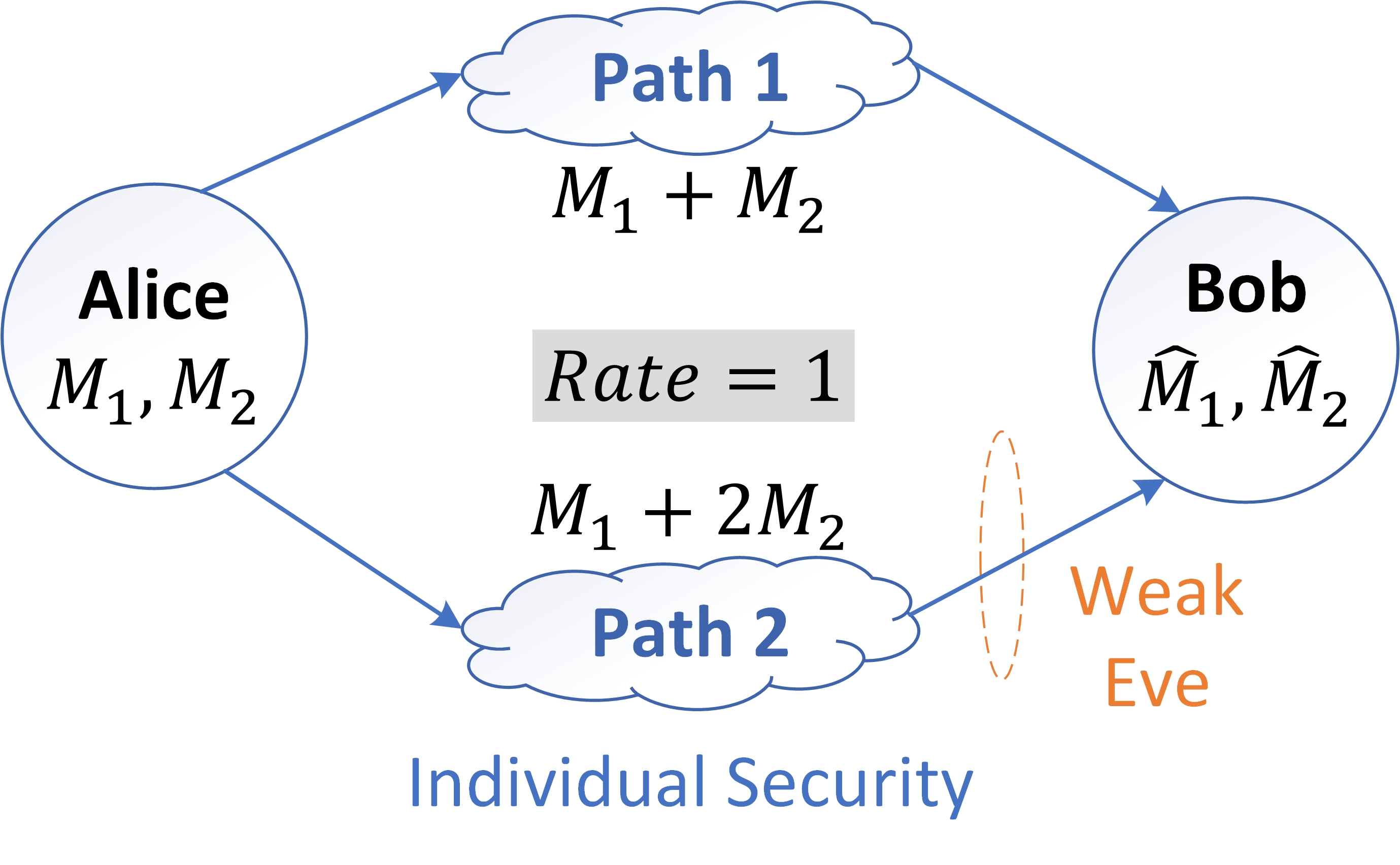}
\caption{}
\end{subfigure}
\begin{subfigure}[b]{0.32\textwidth}
\includegraphics[scale=0.4]{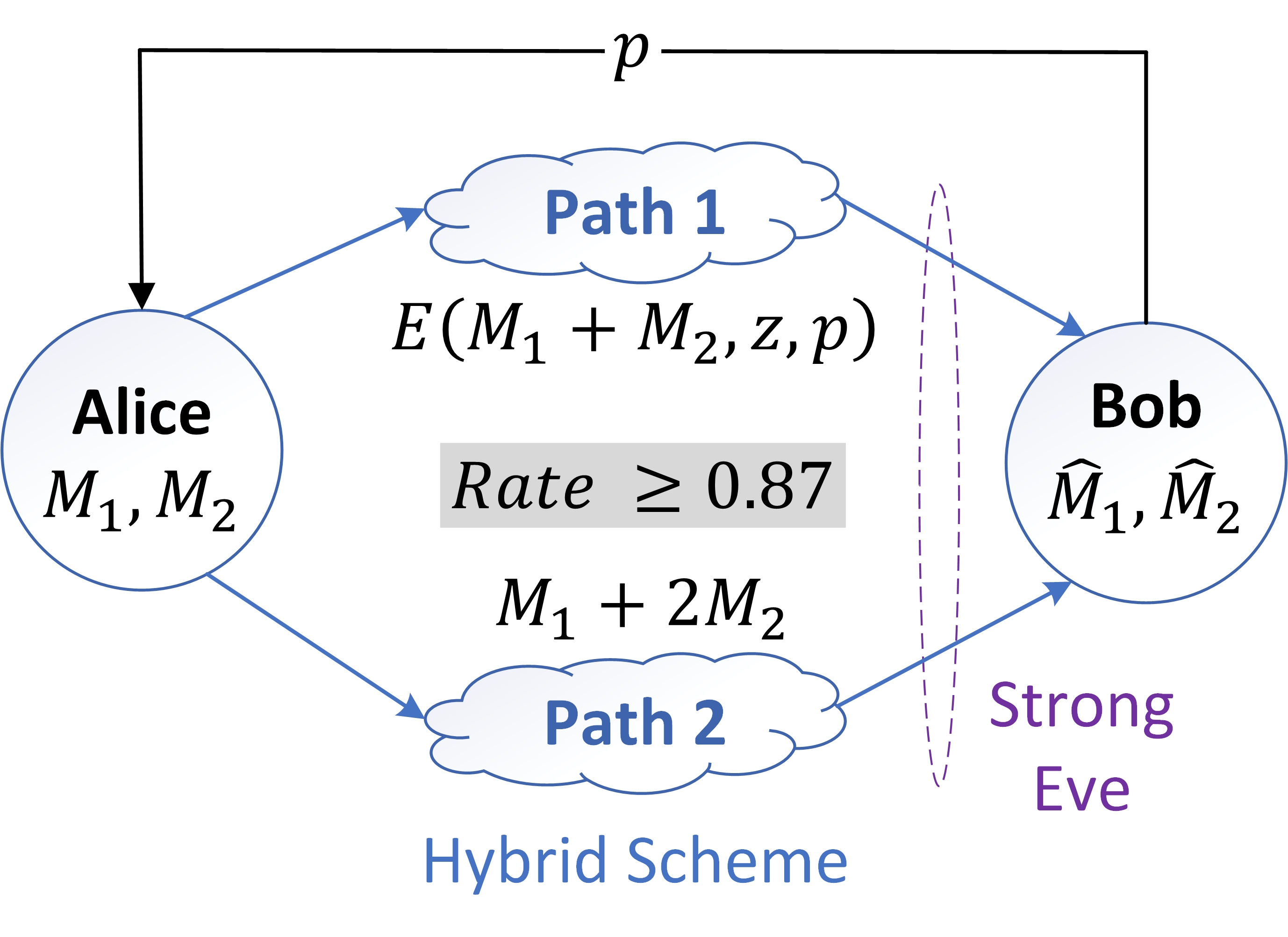}
\caption{} \label{fig:foobar hybrid}
\end{subfigure}
\off{\begin{subfigure}[b]{0.32\textwidth}
\includegraphics[scale=0.4]{./fig/Joint}
\caption{}
\end{subfigure}}
\end{flushleft}
\caption{Secure post-quantum solutions for a multipath network with two paths, one source Alice with two messages $M_1$ and $M_2$, one legitimate destination Bob, and two types of possible eavesdropper's Eve, weak and strong, which can obtain the information transmitted over one or two paths, respectively. (a) One-time pad \cite{shannon1949communication} (b) McEliece cryptosystem \cite{mceliece1978public}, using $[2960,2288]$-Goppa code suggested recently for long term security in \cite{bernstein2008attacking} (c) Network coding wiretap II \cite{el2007wiretap} (d) Individual Security for networks \cite{cohen2018secure} (e) Proposed hybrid universal network-coding cryptosystem (HUNCC), with rate that converges to 1 as $\mathcal{O}\left(1/l\right)$\off{ (f) Proposed joint scheme}.}
\label{fig:foobar}
 \vspace{-3mm}
\end{figure}

\subsection*{Main Contribution}

In this work, we consider a new hybrid network-coding cryptosystem (HUNCC) to obtain post-quantum cryptography at high rates. In this secure network-coding scheme, we combine the computational security principles, with the physical layer security primitives, thus introducing a hybrid system which relies on both individual secrecy, and computational secure cryptosystems.

We illustrate this concept via a multi-path secure transmission scheme, where multiple messages are to be sent between Alice and Bob, using parallel links. By doing so, we are able to address one of the main shortcomings of physical layer security, namely that the adversary eavesdropper cannot observe all the messages sent between Alice and Bob. In the system we will introduce shortly, Eve may in fact observe the entirety of the transmission between Alice and Bob.

We will make an assumption on Eve's computational power however, similar to computationally secure systems.
On the other hand, instead of encrypting the entirety of the messages from Alice using public-key cryptography, we only encrypt over some of the communication links.
Under the computational limitation assumption, the links that are encrypted are analogous to losses for Eve -- and thus, we may now employ traditional techniques from physical layer security codes.
In other words, we are able to enforce that the eavesdropper only observe part of the messages, via cryptography, under the computational limitation assumption. This, in turn, allows an increase in the communication rate -- similar to what is done in individual secrecy, while still providing security guarantees that are computationally strong. If Alice and Bob have more communication channels, the communication rate will be even larger. Indeed, the rate goes to one as the number of channels increases. In Table \ref{tab:perf}, we compare the performance of the individual secure coding scheme and the original McEliece cryptosystem to our proposed hybrid scheme.

\begin{table}[!h]
\centering
\normalsize
\begin{tabular}{c c c c}
\toprule
\multicolumn{1}{c}{}&\multicolumn{1}{c}{Ind. Sec. Code \cite{cohen2018secure}}&\multicolumn{1}{c}{HUNCC}&\multicolumn{1}{c}{McEliece Cryp. \cite{sendrier2002security}}\\
\midrule
Encrypted paths      & $0$                     & $c$                             & $l$            \\
Information  rate            & $R_{IS}=1$                    & $(cR_{C} + (l-c)R_{IS})/l$        & $R_{C}<1$     \\
Public-key size       & $0$                     & $p_{C}$                    & $p_{C}$   \\
Ind. comp. secrecy       & $0$                     & $\min\{c,1\} \cdot b$           & $b$            \\
Ind. secrecy       & $(l-w)/l$               & $(l-w)/l$                           & $0$            \\
\bottomrule
\end{tabular}
\caption{
For a multipath network with $l$ paths we compare individually secure codes, our hybrid code HUNCC, and the McEliece cryptosystem, where in the original McEliece Cryptosystem, the rate $R_C \approx 0.5$. The parameter $c$ corresponds to the number of encrypted path in HUNCC, $w$ is the number of links that a weak eavesdropper observes, and $b$ is the security level of the McEliece Cryptosystem using public-key $p_C$.
}
\label{tab:perf}
\end{table}

In recently considered heterogeneous networks, one may not assume that the cryptosystem, with the public-key, can be applied on all paths. The HUNCC coding scheme can ensure post-quantum security across the entire network, using the public-key only for the information transmitted over one path. Moreover, we show important applications in which HUNCC is applicable. Specifically, we consider single path communication, distributed storage, ultra-reliable low-latency streaming communications, and the case of myopic adversaries.

For a network with two paths, we illustrate our hybrid scheme in Figure \ref{fig:foobar hybrid}: Alice wants to send a private message $M=[M_1,M_2] \in \mathbb{F}_{q^{u}}^2$ to Bob via $l=2$ communication links. An eavesdropper, Eve, depending on how strong she is, can observe the communication in one or both of these channels.

Alice and Bob agree on a public key encryption scheme $(\mathrm{Enc},\mathrm{Dec},p,s)$. Alice first encodes the message $M$ using the individually secure code with generator matrix $\textbf{G} = \left( \begin{smallmatrix} 1 & 1\\  2 & 1 \end{smallmatrix} \right )$. We denote this encoding by $X = M\textbf{G} = [M_1 + M_2 , M_1 + 2 M_2]$. Alice then encodes $\mathrm{Enc}(X_1,p)$, sends it to Bob via channel 1, and sends $X_2$ via channel 2. Bob then decodes $X_1 = \mathrm{Dec}(\mathrm{Enc}(X_1,p),s)$ to retrieve all of $X$ and multiplies by the inverse of the generator matrix to retrieve the original message $M=X\textbf{G}^{-1}$.

If Eve is a weak eavesdropper, observing only one communication channel, the scheme is information-theoretically individually secure irrespective of her computational power. This occurs because each piece of the message, $M_i$, is independent from any single $X_j$. If Eve is a strong eavesdropper, observing both communication channels, each message $M_i$ will be computationally secure with almost the same security level as the encryption scheme $(\mathrm{Enc},\mathrm{Dec},p,s)$. Indeed, we show in Theorem \ref{theo:level_security} that if the best known attack on $(\mathrm{Enc},\mathrm{Dec},p,s)$ needs $2^b$ operations to break it, then Eve needs at least $2^b - \frac{\epsilon}{2^b}$ operations to determine any message $M_i$, where $\epsilon$ is the amount of operations needed to solve a $2\times2$ linear system.\footnote{Using Gaussian elimination would make $\epsilon \approx 10$ operations.}

We note that for the individually secure code $\textbf{G}$ to work, $\mathbb{F}_{q^{u}}$ must have characteristic larger than $2$. Also, the image of the individually secure code must be injectively mapped into the domain of the encryption function $\mathrm{Enc}$. We consider the following example.

Suppose Alice and Bob agree on using a $128$-bit McEliece cryptosystem. Utilizing the suggestion in \cite{bernstein2008attacking}, Bob selects a $[2960,2288]$-Goppa code with a public key of $1537536$ bits\footnote{Original parameters of McEliece cryptosystem \cite{mceliece1978public}, namely $[1024,524]$-Goppa code with a rate of approximate 0.5 obtain around $58$-bit security level considering recent attacks \cite{bernstein2008attacking}.}. In this case, the domain of the encryption function is $\mathbb{F}_{2}^{2288}$. Alice and Bob must agree on an injective mapping from the image of the code $\textbf{G}$, given by $\mathbb{F}_{q^{u}}$ with characteristic larger than $2$, to $\mathbb{F}_{2}^{2288}$. In this case, they could set $q^{u} = 3^{1443}$ so that $\log_2(q^{u}) \approx 2287.1$ bits. Thus, Alice can map $X_1$ into a $2288$ bit vector and encode it using the Goppa code into $E(X_1,p) \in \mathbb{F}_{2}^{2960}$. Alice will then send $\log_2 |E(X_1,p)| = 2960$ bits through link~1 and $\log_2 |X_2| \approx 2287.1$ bits through link~2. Thus, the total communication cost will be around $5248$ bits giving a communication rate slightly larger than $0.87$. By Theorem \ref{theo:level_security} both messages, $M_1$ and $M_2$ are $128$-bit secure. In Section \ref{rsa_example}, we look at an example where Alice and Bob agree on using an RSA scheme.

The structure of this work is as follows. In
Section \ref{sec:background}, we provide a background on the main security blocks we use in the proposed solution. In Section \ref{sec:model}, we
formally describe the system model. The security notions and the threat models we will use in the paper are defined in Section \ref{sec:security_defs}. In Section \ref{sec:hybrid}, we present our scheme HUNCC with our main results. In Section \ref{sec:efficiency} we analyze the performance of HUNCC and discuss our results. In Section \ref{sec:Applications}, we present a few examples for which HUNCC
is applicable. Finally, we conclude the paper in Section \ref{sec: conclusions}.

\section{Background}\label{sec:background}

In this section we give some background information on the two main building blocks of our scheme: i) computational secure cryptosystems and, ii) information-theoretic individual security. We note that, although we focus on the McEliece Cryptosystem, any computationally secure cryptosystem can be used in our scheme. In Section \ref{rsa_example}, we provide an example of how the original RSA can be used in our scheme.

\subsection{Computationally Secure Cryptosystems}\label{sec:crypto}
We start by introducing some helpful notation and definitions for public-key cryptosystems.

\begin{definition}\label{Crypto_scheme}
    A public-key encryption is a tuple $(\mathrm{Enc},\mathrm{Dec},p,s,\kb,\nb)$  where:
    \begin{itemize}
        \item $\mathrm{Enc}: \{0,1\}^{\kb} \times \mathcal{P} \to \{0,1\}^{\nb}$ is the encryption function, $\mathrm{Dec}: \{ 0,1\}^{\nb} \times \mathcal{S} \to \{0,1\}^{\kb}$ is the deciphering function, and $p \in \mathcal{P}$ and $s \in \mathcal{S}$ represent the public and private key respectively,
        \item For a message $\bar{m} \in \{0,1\}^{k_b}$, a public key $p \in \mathcal{P}$, and the corresponding secret key $s \in \mathcal{S}$, $\mathrm{Dec}(\mathrm{Enc}(\bar{m}, p),s) = \bar{m}$,
    \end{itemize}
    A public-key encryption $(\mathrm{Enc}, \mathrm{Dec},p,s,\kb,\nb)$ has security level $b$ if the best known algorithm to recover $\bar{m}$ with the knowledge of $\mathrm{Enc}(\bar{m})$ and $p$ alone needs to perform at least $2^b$ operations. Finally $(\mathrm{Enc}, \mathrm{Dec},p,s,\kb,\nb)$ is said to have rate $R$ if $R = \kb / \nb$.
\end{definition}
Parameters for long-term security in post-quantum cryptography were suggested in \cite{bernstein2008attacking,PQCRYPTO2015}. These parameters determine a trade-off between the security level, the information rate, and the size of the public-key in the cryptosystems. We focus now on the McEliece cryptosystem and take a look at RSA in Section \ref{rsa_example}.

\subsubsection{The McEliece Cryptosystem}
The original McEliece cryptosystem \cite{mceliece1978public} remains unbroken. In this scheme, Bob produces a generator matrix $\textbf{G} \in \mathbb{F}_{q}^{\kb \times \nb}$ by randomly choosing an irreducible polynomial of degree $t$ over $GF(2^d)$, which corresponds to an irreducible Goppa code of length $\nb = 2^d$ and dimension $\kb \geq \nb - td$. Note that this code can correct at least $t$ errors\footnote{We note that in \cite[Section 7]{bernstein2008attacking}, they consider $\nb - \sqrt{\nb(\nb-2t-2)}\geq t+1$ errors instead of $t$
errors.} and can be decoded efficiently \cite{bernstein2009introduction}. Bob then generates two matrices with the intent of concealing $\textbf{G}$, a random dense nonsingular matrix $\textbf{S} \in \mathbb{F}_{q}^{\kb \times \kb}$ and a random permutation matrix $\textbf{P} \in \mathbb{F}_{q}^{\nb \times \nb}$. The scheme works as follows.

\noindent \underline{Key Generation}:
\begin{itemize}
    \item \underline{Public Key}: Bob generates the public key $(\textbf{G}^{\text{pub}} = \textbf{SGP} ,t)$ , where $\textbf{G}^{\text{pub}} \in \mathbb{F}_{q}^{\kb \times \nb}$. Both Alice and Eve have access to it.
    \item \underline{Private Key}: The private key consists of $(\textbf{S}, D_{\mathcal{G}}, \textbf{P})$, where $D_{\mathcal{G}}$ is an efficient decoding algorithm for $\mathcal{G}$.
\end{itemize}

\noindent \underline{Encryption}: To encrypt a message $\bar{\bold{m}}\in F^{\kb}_{q}$, Alice randomly chooses a vector $\bold{z}\in F^{\nb}$ of weight $t$ and encrypts it as  $\bold{c} = \bar{\bold{m}}\textbf{G}^{\text{pub}} \oplus \bold{z}$.

\noindent \underline{Decryption}: To decrypt the message, Bob first calculates $\bold{c}\textbf{P}^{-1} = \bar{\bold{m}}\textbf{SG} \oplus \bold{z}\textbf{P}^{-1}$, and then applies the decoding algorithm $D_{G}$. Since $\bold{c}\textbf{P}^{-1}$ has hamming distance $t$, it follows that, $\bar{\bold{m}}\textbf{SG} = D_{G}(\bold{c}\textbf{P}^{-1})$. Then, since both $\textbf{G}$ and $\textbf{S}$ are invertible, $\bar{\bold{m}} = (\bar{\bold{m}}\textbf{SG}) \textbf{G}^{-1} \textbf{S}^{-1}$.

Many attempts have been made to improve the performance of the original McEliece cryptosystem by replacing Goppa codes with other families of codes. Most of these, however, have been broken. We summarize some of these attempts in Table \ref{tab:code}.

\begin{table}[!h]
\centering
\normalsize
\begin{tabular}{c c c}
\toprule
\multicolumn{1}{c}{Family Code}&\multicolumn{1}{c}{Proposed by}&\multicolumn{1}{c}{Broken by}\\
\midrule
Goppa                   & McEliece, 1978 \cite{mceliece1978public}           &                   - \\
Reed Solomon            & Niederreiter, 1986 \cite{niederreiter1986knapsack} & Sidelnikov et al, 1992 \cite{sidelnikov1992insecurity} \\
Concatenated            & Niederreiter, 1986 \cite{niederreiter1986knapsack} & Sendrieret al, 1998 \cite{sendrier1998concatenated}    \\
Read-Muller             & Sidelnikov, 1994 \cite{sidelnikov1994public}       & Minder et al, 2007 \cite{minder2007cryptanalysis}      \\
LDPC                    & Monico et al, 2000 \cite{monico2000using}          & Monico et al, 2000 \cite{monico2000using}              \\
Convolutional codes     & L\"{o}ndahl et al, 2012 \cite{londahl2012new}      & Landias et al, 2013 \cite{landais2013efficient}        \\
\bottomrule
\end{tabular}
\caption{Code-based encryption systems.}
\label{tab:code}
\end{table}

\subsection{Individual Information-theoretic Security}\label{sec:IS}

Individual security operates under the assumption that Eve is a weak eavesdropper $Y_{E_w}$, i.e. only has access to $w$ communication paths. The privacy guarantee is that, while Bob is able to decode completely all the $\ku \in \mathbb{F}_{q^{u}}$ messages transmitted over the network of length $\kb \in \mathbb{F}_{2}$ bits each, Eve is ignorant with respect to each individual message. Thus,
\[
    H(M_j|Y_{E_{w}}) = H(M_j) \text{ for all } j \in \{1, \ldots , \ku \}.
\]

In general, individual security does not imply perfect privacy over the entire messages transmitted. Eve may potentially obtain information about mixtures of the encrypted messages transmitted. In many applications, this information is considered insignificant \cite{bhattad2005weakly,silva2009universal,silva2011universal,mansour2014secrecy,chen2015individual,mansour2015individual,mansourindividual, goldenbaum2015multiple,chensecure, mansour2015individual1,cohen2018secure}. If the messages are statistically independent, this means that,
\[
    I(M_{\ku};Y_{E_{w}}|M_1,\ldots,M_{\ku-1})\\ = H(M_{\ku}) - H(M_{\ku}|Y_{E_{w}},M_1,\ldots,M_{\ku-1}) \geq I(M_{\ku};Y_{E_{w}}).
\]
In \cite{cohen2018secure}, individually security is generalized to $k_s$-individual perfect secrecy, where $\ks \leq \ku-w$, such that the eavesdropper has zero mutual information with any set of $\ku-w$ messages i.e.,
\begin{equation}\label{level_is}
    I(M^{\ku-w},Y_{E_{w}}) = 0 .
\end{equation}
And in \cite{matsumoto2017universal}, the assumption that the messages be independently and uniformly distributed, is removed.

The authors in \cite{cohen2018secure} present two code constructions which satisfy \eqref{level_is}, a random code over a binary field, and a $(\ku,w)$-linear code over a field $\mathbb{F}_{q^u}$ with $u \geq \ku$, where
\[
    \ku \geq \Bigg\lceil \frac{l}{l-w}\Bigg\rceil\geq 2,
\]
and the number of paths, $l$, is greater or equal to $\ku$.
We now show a detailed code construction of the linear code.

\noindent \underline{Code Generation}:
Let $\mathcal{C}$ be a linear code over $\mathbb{F}_{q^u}$, with $u\geq \ku$, of length $\ku$ and dimension $w$, and set $\ks = \ku-w$. Let $\textbf{G}_{IS}^{\star\star} \in \mathbb{F}_{q^u}^{w \times \ku}$ be a generator matrix for $\mathcal{C}$ and $\textbf{G}^{\star}_{IS} \in \mathbb{F}_{q^u}^{\ks \times \ku}$ a generator for the null space of $\mathcal{C}$. Finally, let $\textbf{H}_{IS} \in \mathbb{F}_{q^u}^{\ks \times \ku}$ and $\tilde{\textbf{G}}_{IS} \in \mathbb{F}_{q^u}^{w \times \ku}$, be the parity check matrix and the basis matrix for the code $\mathcal{C}$, respectively, i.e. such that, $\textbf{H}_{IS}\textbf{G}^{\star T}_{IS}=\textbf{I}$ and  $\tilde{\textbf{G}}_{IS}\textbf{G}_{IS}^{\star\star}=\textbf{I}$ yet $\tilde{\textbf{G}}_{IS}\textbf{G}^{\star}_{IS}=0$. Then, the individual security code is generated by $\left[ \begin{smallmatrix} \textbf{G}^{\star}_{IS} \\ \textbf{G}_{IS}^{\star\star} \end{smallmatrix} \right] \in \mathbb{F}_{q^u}^{\ku \times \ku}$.

It is important to note that unlike in public-key cryptosystems, and in particular the McEliece cryptosystem, in physical-layer security schemes, the generation matrix and the code is public. Thus, we can assume that both Bob and Eve have access to all the matrices described above.

\noindent \underline{Encryption}:
Alice encrypts the message $M \in \mathbb{F}_{q^u}^{\ku \times 1}$ as $X^T = M^T \left[ \begin{smallmatrix} \textbf{G}^{\star}_{IS} \\ \textbf{G}_{IS}^{\star\star} \end{smallmatrix} \right] \in \mathbb{F}_{q^u}^{1 \times \ku}$.

\noindent \underline{Decryption}: To decrypt the message, Bob uses the parity check matrix $\textbf{H}_{IS}$ and the basis matrix $\tilde{\textbf{G}}_{IS}$ to compute $(M_{1};\ldots;M_{k^{\prime}}) = \textbf{H}_{IS} X$ and $(M_{k^{\prime}+1};\ldots;M_{k}) = \tilde{\textbf{G}}_{IS} X$. By \cite{cohen2018secure}, since Eve only observes $w$ symbols from the encrypted vector $X^T$ she is not able to decode and is completely ignorant with respect to any set $k-w$ symbols of information transmitted over the network.

\section{System Model}\label{sec:model}
We consider a network consisting of a source node, Alice, connected to a destination node, Bob, via $l$ noiseless independent communication links. The network is illustrated in Figure~\ref{fig:system_model}. The goal is for Alice to transmit $\ku$ messages $M=[M_1;\ldots;M_{\ku}] \in F^{\ku}_{q^u}$, of length $\kb \in \mathbb{F}_{2}$ bits each, privately to Bob in the presence of an eavesdropper, Eve. In this paper, our main focus is on post-quantum security. Thus, we assume Eve has access to a quantum computer.

We denote by $Y = [Y_1;\ldots;Y_l]$ the vector of messages which Alice sends to Bob via each communication link. These messages must be such that Bob is able to decode $M$ from them. Thus, we say $Y$ is reliable if $H(M | Y) = 0$. The messages $Y$, however, should satisfy certain properties if Eve is not to decode $M$ herself. This will depend on how powerful Eve is.

We consider two types of Eve. A strong Eve, $E_s$, which observes all communication links, having access to the entirety of $Y$. And a weak Eve, $E_w$, which only observes a subset of the communication links. We denote the observations of each by $Y_{E_s}$ and $Y_{E_w}$, respectively. We note that because of reliability, information-theoretic privacy is not obtainable against a strong Eve.

\section{Security Definitions}\label{sec:security_defs}
In this section, we define the security notions and the threat models used in the rest of the paper.
Throughout, we assume a \emph{ciphertext-only attack} model, i.e., the adversary Eve only has access to $Y_{E_s}$ or $Y_{E_w}$. Standard techniques from cryptography allow to expand the threat model to chosen-plaintext attacks, see e.g. \cite{biryukov1998differential}.
We first define the notion of computational security.

\begin{definition}\label{def:sec_level}
A cryptosystem with message $M$ and ciphertext $c(M)$ has security level $b$ if the best known algorithm needs to perform at least $2^b$ operations to decode $M$ from the observation of $c(M)$ alone in expectation\footnote{The expectation is taken over the distribution of $X$, and in the case that the best decoding algorithm is a randomized algorithm, also taken over the randomized algorithm. Also note that while this notion of security is in expectation, it can be generalized to a stronger notion, e.g. that the algorithm takes less at least $2^b$ operations with probability at least $1-\epsilon$, for some small $\epsilon$, see \cite{micciancio2018bit}.}.
\end{definition}

It should be noted that this definition of security level imposes implicitly a restriction on the size of the encoded message $M$, and on its distribution. Indeed, assuming a public-key cryptosystem, the adversary may always use a \emph{brute-force attack}, by guessing potential message inputs until the correct one is found. Generally, it is assumed that the message $M$ takes a value uniformly at random in a set $\{0,1\}^{\kb}$, and thus it follows that $\kb \geq b$. In fact, most cryptosystems require $\kb$ to be strictly larger than $b$. Throughout this paper, we opt to leave this relationship implicit, and rather focus on the security level $b$, and select the appropriate $\kb$ based on this security level.

The computation security level of Definition~\ref{def:sec_level} is relevant in the case of a single link, or equivalently a single message.
When there are several messages being transmitted on each link, it is desirable to provide a security guarantee that applies to each message \emph{individually}. For a weak Eve, this can be obtained via information-theoretic individual secrecy.

\begin{definition}\label{individually}
A cryptosystem with messages $M_1,\ldots, M_l$ is $(l,w)$-individually secure if for every $\omega \subset [l]$, with $|\omega| = w$, $H(M_s | Y_{\omega}) = H(M_s)$ where $Y_\omega = [M_i]_{i \in \omega}$.
\end{definition}

\begin{remark}
It should be noted that the linear code construction given in Section~\ref{sec:IS} actually guarantees a stronger notion of individual secrecy, c.f. \eqref{level_is}. While such distinction is important -- it reduces the possible set of joint messages $M_1,\ldots,M_l$ exponentially -- the notion of individual secrecy in Definition~\ref{individually} is more convenient to manipulate, and turns out to be sufficient to prove our main results.
\end{remark}

For a strong Eve, which can observe the entirety of the sent messages, $(l,w)$-individual security is unobtainable. Instead, we propose a notion of \emph{individual computational security}, which states that the decoding of any message on any of the paths would require $2^b$ operations.

\begin{definition}[Individual Computational Secrecy]\label{def:sec_level_ind}
A cryptosystem with messages $M_1,\ldots,M_l$ and ciphertexts $c_i \triangleq c_i(M_1,\ldots,M_l)$, for $i = 1,\ldots, l$ has security level $b$ if the best known algorithm needs to perform at least $2^b$ operations to decode any $M_j$, $j = 1,\ldots,l$ from the observations $\mathbf{c}_1,\ldots,\mathbf{c}_l$, in expectation.
\end{definition}

Note that individual computational secrecy implies in the computational security of the cryptosystem with $M = [M_1,\ldots,M_l]$ and $c(M) = [c_1(M_1,\ldots,M_l),\ldots, c_l(M_1,\ldots,M_l)]$. It is therefore a strictly stronger notion of security.

\section{Hybrid Universal Network-Coding Cryptosystem (HUNCC)}\label{sec:hybrid}
\begin{figure}[t]
    \centering
    \includegraphics[scale=0.5]{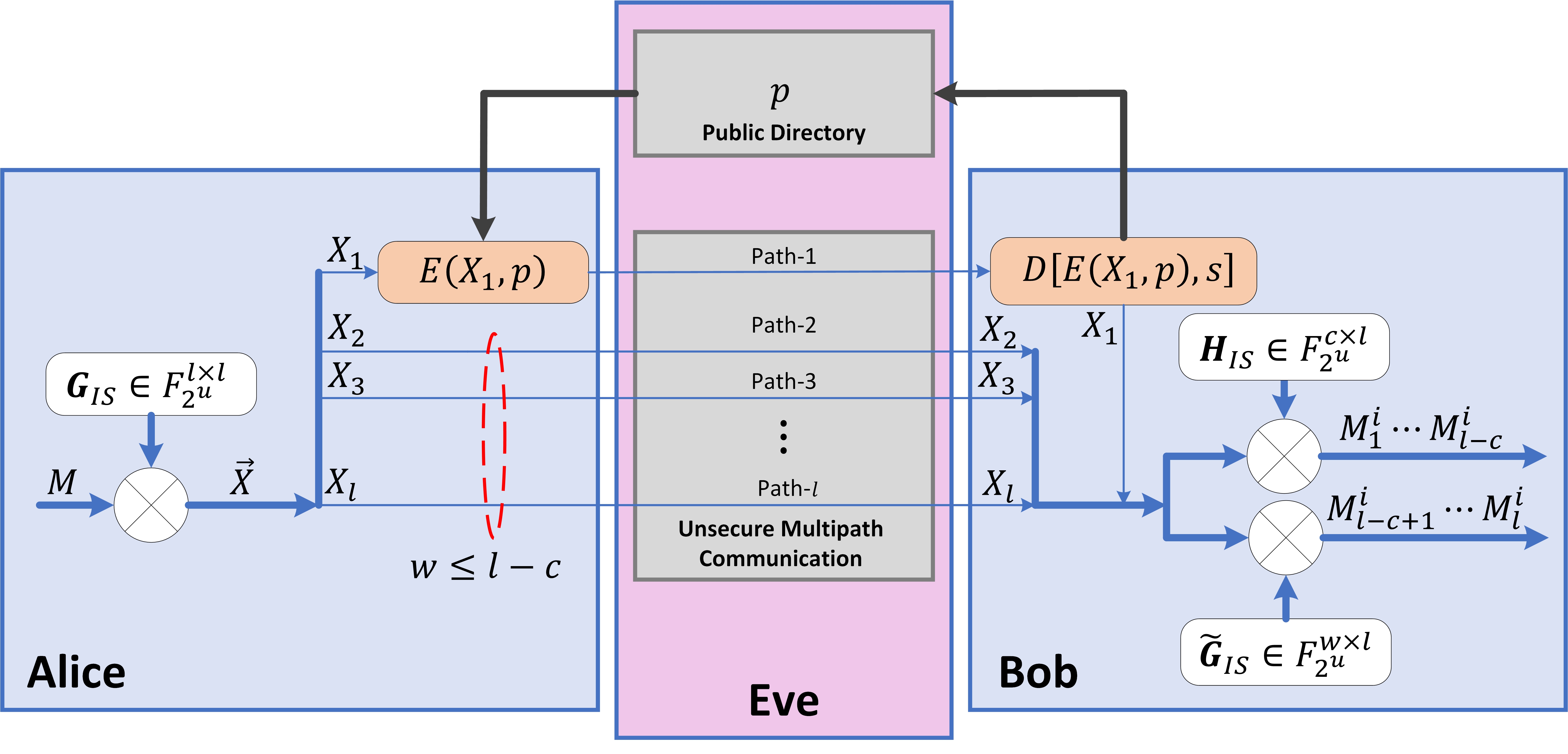}
    \caption{Hybrid post-quantum secure multipath network scheme.}
    \label{fig:Genrral_Scheme}
    \vspace{-0.2in}
\end{figure}

In this section, we present our proposed hybrid cryptosystem, which we refer to as Hybrid Universal Network-Coding Cryptosystem (HUNCC). In Figure \ref{fig:Genrral_Scheme}, we illustrate how HUNCC operates on a multipath network with $l$ communication links.

At a high level, the cryptosystem works as follows. Alice and Bob agree on a public-key cryptosystem $(\mathrm{Enc},\mathrm{Dec},p,s,\kb,\nb)$ as in Definition \ref{Crypto_scheme}. Alice then selects $c$ of the $l$ links to be encrypted links. This alone does not provide full security, as the $l - c$ links which do not experience encryption are unsecured. To solve this, we introduce an extra step before the encryption phase in which the messages are mixed using an individually secure linear code as in Section \ref{sec:IS}. By doing this, we show in Theorem \ref{theo:level_security}, that each message is now computationally secure. The security performance and the total communication rate of the scheme will depend on the choice of $c$ and of the parameters of the public-key cryptosystem.

\begin{figure}[t]
    \centering
    \includegraphics[scale=0.63]{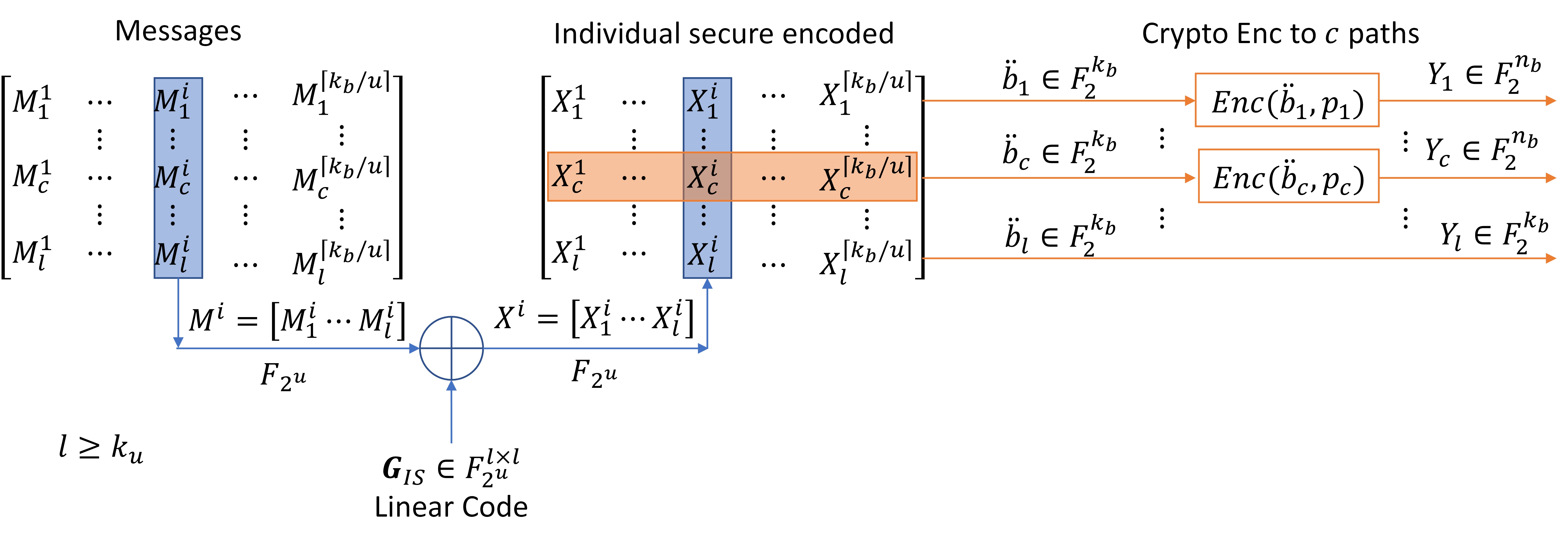}
    \caption{HUNNC encoding scheme at Alice.}
    \label{fig:Alice_enc}
    \vspace{-0.2in}
\end{figure}


\subsection{Scheme Description}

We now give a more detailed explanation of HUNCC, as presented in Algorithm \ref{alg:hybrid_scheme}. We start by looking at the encoding process at Alice, illustrated in Figure~\ref{fig:Alice_enc}. Let $(\mathrm{Enc},\mathrm{Dec},p,s,\kb,\nb)$ be a public-key cryptosystem, as described in Definition \ref{Crypto_scheme}, with security level $b$. As usual, the public key is generated by Bob and provided to Alice over some public communication channel. Alice chooses a number $c$ of paths to be encrypted. Without loss of generality, we let the paths indexed by $1,\ldots,c$ to be the encrypted ones. Let $u \geq l$, be fixed, and consider blocks of messages $M^{(1)}, \ldots,M^{(\lceil \kb / u \rceil )}$, where each $M^{(i)} = [M_1^{(i)}, \ldots, M_{l}^{(i)}]$ with $M_j^{(i)} \in \mathbb{F}_{2^u}$ is generated independently, and uniformly at random. Let $\textbf{G}_{\mathrm{IS}}^{\star\star} \in \mathbb{F}_{2^u}^{w \times l}$ be an $(l,w)$-individually secure linear code\footnote{For practical consideration, if the number of the paths $l$ is high, such that one receives more than $\kb$ bits representing $l$ in the binary field, it is possible to divide the set of messages, $\ku$, Alice encodes together and perform the same process for each set. This division can be used as well to reduce the field size of the individual security code.} as in Section \ref{sec:IS}, where $w \leq l-c$. Select a generator matrix $\textbf{G}_{\mathrm{IS}}^{\star} \in \mathbb{F}_{2^u}^{c \times l}$ for the null space of the code. Note that since $u \geq l$ and $l-w \geq 1$, such a code exists. The individual secrecy encoding matrix is then given by $\textbf{G}_{IS} = \left[ \begin{smallmatrix}   \textbf{G}^{\star}_{IS} \\ \textbf{G}_{IS}^{\star\star} \end{smallmatrix} \right] .$ Thus, the vectors $X^{(1)}, \ldots, X^{(\lceil \kb / u \rceil )}$, where $X^{(i)} = M^{(i)} \textbf{G}_{\mathrm{IS}}$ correspond to an $(l,w)$-individual secrecy encoding of $M^{(i)}$ (c.f. lines 1-5 in Algorithm~\ref{alg:hybrid_scheme}).

Now, for every path $i \in [c]$, consider the collection of symbols $X_{i}^{(1)},\ldots, X_{i}^{(\lceil \kb / u \rceil)}$. Since $\mathbb{F}_{2^u} \simeq \mathbb{F}_2^{u}$, the collection $X_{i}^{(1)},\ldots, X_{i}^{(\lceil \kb / u \rceil)}$ can be injectively mapped into a sequence of bits $\ddot{\mathbf{b}}_{i}$ of length $\kb$. Each $\ddot{\mathbf{b}}_i$ is encoded via the public-key encryption before being sent, i.e., each link $i$ transmits $\mathbf{y}_i = \mathrm{Enc}(\ddot{\mathbf{b}}_i, p_i)$. Note that $\mathbf{y}_i$ is of length $n$ (c.f. lines 6-13 in Algorithm~\ref{alg:hybrid_scheme}).

For the paths $i > c$, Alice directly sends the collection $X_{i}^{(1)},\ldots, X_{i}^{(\lceil \kb / u \rceil)}$ unencrypted. For consistency, it will be easier to also assume that the data sent on the path is a bit sequence $\ddot{\mathbf{b}}_i$, this time of length $\kb$ (c.f. lines 14-16 in Algorithm~\ref{alg:hybrid_scheme}).

\begin{algorithm}
  \caption{Hybrid Universal Network-Coding Cryptosystem (HUNCC)}
  \label{alg:hybrid_scheme}
  \begin{algorithmic}[1]
    \Statex\textbf{Input:}  At Alice, $\ku$ messages $M=[M_1;\ldots;M_{\ku}] \in F^{\ku}_{q^u}$ of length $\kb \in \mathbb{F}_{2}$ bits each

    \Statex\textbf{\underline{Encoding scheme at Alice}:}
    \State \textbf{Stage 1:} Individual secrecy encoding
    \State For $u\geq l$, represent the $\ku$ messages in a block $[M^{(1)}, \ldots, M^{(\lceil \kb / u \rceil )}] \in \mathbb{F}_{2^{u}}$
    \ForEach{column $i\in\{1,\ldots,\lceil \kb / u \rceil\}$ in the block  $[M^{(1)}, \ldots,M^{(\lceil \kb / u \rceil )}]$}
        \State $X^{(i)} = M^{(i)} \textbf{G}_{\mathrm{IS}}$
    \EndFor
    \State \textbf{Stage 2:} Public-key encryption
    \ForEach{path $1\leq i \leq c$ which support cryptosystem}
        \State $\ddot{\mathbf{b}}_{i}\in \mathbb{F}_2^{\kb} \leftarrow [X_{i}^{(1)},\ldots, X_{i}^{(\lceil \kb / u \rceil)}] \in \mathbb{F}_{2^{u}}$
        \State $\mathbf{y}_i = \mathrm{Enc}(\ddot{\mathbf{b}}_i, p_i)$
    \EndFor

    \Statex\textbf{\underline{Transmission over the multipath network}:}
    \ForEach{path $1\leq i \leq c$ which support cryptosystem}
        \State Transmit $\mathbf{y}_i$
    \EndFor
    \ForEach{remaining path $1\leq i \leq (l-c)$ without cryptosystem}
        \State Transmit $\mathbf{y}_i = [X_{i}^{(1)},\ldots, X_{i}^{(\lceil \kb / u \rceil)}] \in \mathbb{F}_{2^{u}}$ (if required $\mathbf{y}_i \in \mathbb{F}_2^{\kb} \leftarrow \mathbf{y}_i \in \mathbb{F}_{2^{u}}$)
    \EndFor

    \Statex\textbf{\underline{Decoding scheme at Bob}:}
    \State \textbf{Stage 1:} Public-key decryption
    \ForEach{path $1\leq i \leq c$ with cryptosystem}
        \State $[X_{i}^{(1)},\ldots, X_{i}^{(\lceil \kb / u \rceil)}] \leftarrow \mathrm{Dec}(\mathbf{y}_i, s_i)$
    \EndFor
    \State \textbf{Stage 2:} Individual secrecy decoding
    \ForEach{column $i\in\{1,\ldots,\lceil \kb / u \rceil\}$ in the block  $[X^{(1)}, \ldots,X^{(\lceil \kb / u \rceil )}]$}
        \State $(M_{1}^{i};\ldots;M_{l-c}^{i}) = \textbf{H}_{IS} X^{i}$
        \State $(M_{l-c+1}^{i};\ldots;M_{l}^{i}) = \tilde{\textbf{G}}_{IS}X^{i}$.
    \EndFor
    \Statex\textbf{Output:} At Bob, $M=[M_1;\ldots;M_{\ku}] \in F^{\ku}_{q^u}$
  \end{algorithmic}
\end{algorithm}

We now detail the decoding process at Bob. As per Section \ref{sec:model}, we assume that all paths are error-free. Hence, Bob obtains all the messages transmitted over the network. For the $c$ encrypted paths, Bob uses the  private-key to decode the messages (c.f. lines 17-20 in Algorithm~\ref{alg:hybrid_scheme}). Thus, Bob obtains $[X_{i}^{(1)},\ldots, X_{i}^{(\lceil \kb / u \rceil)}] = \mathrm{Dec}(\mathbf{y}_i, s_i)$ for every $i \in [c]$. The messages obtained via the remaining $l-c$ paths were unencrypted. Thus, together with the decrypted messages from the first $c$ paths, Bob has the entirety of $X^{(1)}, \ldots,X^{(\lceil \kb / u \rceil )}$. Now, for each $i$-th column in the estimated encoded block, Bob uses the parity check $\textbf{H}_{IS} \in \mathbb{F}_{2^u}^{c \times l}$ and the basis matrix $\tilde{\textbf{G}}_{IS} \in \mathbb{F}_{2^u}^{w \times l} $, as defined in Section \ref{sec:IS}, to obtain the original messages transmitted (c.f. lines 21-25 in Algorithm~\ref{alg:hybrid_scheme}), i.e. $(M_{1}^{i};\ldots; M_{l-c}^{i}) = \textbf{H}_{IS}X^{i}$
and $(M_{l-c+1}^{i};\ldots; M_{l}^{i}) = \tilde{\textbf{G}}_{IS}X^{i}$.

\subsection{Security and Information Rate}

We prove that HUNCC is computationally secure, as per Definition~\ref{def:sec_level_ind}. This makes HUNCC well suited for dealing with a strong Eve which observes the entirety of the communication links. The proof relies on the following elementary lemma, which is the key interface between the computational security of the public-key cryptosystem, and the individual secrecy of the linear code.

\begin{lemma}\label{lem:lem1}
Let $f: \mathcal{X} \to \mathcal{V}$ be an injective function with computational security level $b$, i.e., the inversion function $f^{-1}(f(X)) = X$ takes at least $2^b$ operations in expectation to complete. Consider an arbitrary pair of random variables $(X,Y)$, such that $X \in \mathcal{X}$ and $Y \in \mathcal{Y}$ are independent. Then, any function $g: \mathcal{V} \times \mathcal{Y} \to \mathcal{X}$ such that $g(f(X),Y) = X$ takes at least $2^b$ operations to complete in expectation.
\end{lemma}

\begin{proof}
The proof is by contradiction. Denote by $\mathrm{Comp}(h(x))$ the number of operations that the best algorithm for computing a function $h(x)$ assumes. Assume that there exists a function $g$ such that $g(f(X),Y) = X$ with $\mathbb{E}_{X,Y}[\mathrm{Comp}(g(f(X),Y))] < 2^b$. We have:
\begin{align*}
    \mathbb{E}_{X,Y}[\mathrm{Comp}(g(f(X),Y))] & = \mathbb{E}\left[\mathbb{E}[\mathrm{Comp}(g(f(X),Y))|Y] \right] \\
    & = \sum_{y \in \mathcal{Y}} P_{Y}(y)\mathbb{E}[\mathrm{Comp}(g(f(X),y))] < 2^b.
\end{align*}
Thus, there must exist a specific $y_0 \in \mathcal{Y}$ such that $\mathbb{E}[\mathrm{Comp}(g(f(X),y_0))] < 2^b$. But then, the function $h: \mathcal{V} \to \mathcal{X}$ defined as $h(v) = g(v,y_0)$ completes in less than $2^b$ operations, which is a contradiction.
\end{proof}
Lemma~\ref{lem:lem1} essentially states that an independent and unrelated observation $Y$ may not help in decoding more efficiently a ciphertext $f(X)$.

We are now ready to prove the main theorem of this section.

\begin{theorem}\label{theo:level_security}
Let $u \geq l$, $c \geq 1$, and $(\mathrm{Enc},\mathrm{Dec},p,s,\kb,\nb)$ be a public-key cryptosystem  with security level $b$. Then, the cryptosystem with input $M=[M_1; \ldots; M_l]\in F^{\ku}_{q^u}$ in Algorithm~\ref{alg:hybrid_scheme} is individually computationally secure with level at least $b - \delta / 2^{b}$, where $\delta$ is the number of operations needed to solve an $l\times l$ linear system of equations, i.e. $\delta = \mathcal{O}(l^3)$.
\end{theorem}

\begin{proof}
The proof is in two parts. First, we show that recovering $\ddot{\mathbf{B}}_i$, for $i = 1,\ldots,c$ requires at least $2^b$ operations.
Next, we show that there cannot be a decoding algorithm $g(Y_1,\ldots, Y_l) = M_j$ which completes with too few operations, as otherwise, such decoding could be used to decipher $\ddot{\mathbf{B}}_i$ more efficiently.
For ease of notation, we let $c = 1$, i.e., only one path is being encoded, and $l = 2$, i.e. there are only two paths in total. The general case follows mutadis mutandis.

\textbf{Part 1: Recovering $\ddot{\mathbf{B}}_i$ is hard:} We start by showing that it takes at least $2^b$ operations to recover $\ddot{\mathbf{B}}_1$.
Let $g(\mathbf{Y}_1,\mathbf{Y}_2) = \ddot{\mathbf{B}}_1$ be any decoding function that Eve employs.
We have:
\begin{align*}
    H(\ddot{\mathbf{B}}_1 | \mathbf{Y}_2) &= H(X^{(1)}_1, \ldots, X^{(\kb/u)}_1 | X^{(1)}_2, \ldots, X_2^{(\kb/u)}) \\
    & = \sum_{j = 1}^{k/u} H(X_1^{(j)}|X_2^{(j)}) \\
    & = \sum_{j = 1}^{k/u} H(X_1^{(j)})\\
    & = H(\ddot{\mathbf{B}}_1),
\end{align*}
where the penultimate equality follows from the individual secrecy assumption on $\textbf{G}_{\mathrm{IS}}$. Therefore, we have that $\ddot{\mathbf{B}}_1$ is independent of $\mathbf{Y}_2$.
Since $\mathbf{Y}_1 = \mathrm{Enc}(\ddot{\mathbf{B}}_1)$, we thus have that the decoding function of Eve $g(\mathrm{Enc}(\ddot{\mathbf{B}}_1),\mathbf{Y}_2)$ falls under the setting of Lemma~\ref{lem:lem1}, and thus requires at least $2^b$ operations to complete since $(\mathrm{Enc},\mathrm{Dec},p,s,\kb,\nb)$ has security level $b$.

\textbf{Part 2: Recovering $M_j$ cannot be too easy:}
Let $g(Y_1,Y_2) = M_j = [M_{j}^{1}, \ldots, M_{j}^{[\kb/u]}]$ complete in $2^{\tilde{b}}$ operations for $j = 1, 2$.
Since $Y_2 = \ddot{\mathbf{B}}_2$ can be written equivalently as $[X_2^{1}, \ldots, X_{l}^{[\kb/u]}]$, we can construct the matrix:
\begin{align*}
    A &= \left[ \begin{array}{ccccc}
        M_{j}^1 & \ldots & M_{j}^{i} & \ldots & M_{j}^{[\kb/u]} \\
        X_2^{1} & \ldots & X_{2}^{i} & \ldots & X_{2}^{[\kb/u]}
    \end{array}\right]
\end{align*}
Next, using Lemma~\ref{lem:linear_codes_MDS}, there exist a matrix $\tilde{A}$, which can be constructed from $\mathbf{G}_{\mathrm{IS}}$ such that:
\begin{align*}
    \tilde{A} A &= \left[ \begin{array}{ccccc}
        X_{1}^1 & \ldots & X_{j}^{i} & \ldots & X_{1}^{[\kb/u]} \\
        X_2^{1} & \ldots & X_{2}^{i} & \ldots & X_{2}^{[\kb/u]}
    \end{array}\right]
    =\left[ \begin{array}{c}
        X_1\\
        X_2
    \end{array}\right]
\end{align*}
Note that, in general, performing this matrix multiplication takes no more than $\delta = \mathcal{O}(l^3)$ operations\footnote{Precise constants can be obtained, and depend on the field $F_{2^u}$ over which the matrix inversion is performed. In any case, these constants are negligible when looking at the overall computational security for $b$ in the order of 50 or more.} as we will show in Lemma~\ref{lem:linear_codes_MDS}.
Therefore, using once again the fact that $\ddot{\mathbf{B}}_1$ can be obtained equivalently from $[X_1^{1},\ldots, X_{1}^{[\kb/u]}]$, we just constructed an Algorithm to recover $\ddot{\mathbf{B}}_1$ from the observations $\mathbf{Y}_1,\mathbf{Y}_2$ using $2^{\tilde{b}} + \delta$ operations.
It follows from Part~1 that $\log(2^{\tilde{b}} + \delta) \geq b$, which implies in $\tilde{b} \geq \log(2^b - \delta)$.
Finally, via Taylor expansion, when $\delta$ is sufficiently small, we get:
\begin{align*}
    \log(2^{b} - \delta) = b + \log(1 - \delta / 2^{b}) \approx b - \delta / 2^b.
\end{align*}
Therefore, we get that recovering $M_j$ must take at least $b - \delta / 2^{b}$ operations.
\end{proof}

We now prove Lemma \ref{lem:linear_codes_MDS}, used in the proof of the Theorem above.

\begin{lemma}\label{lem:linear_codes_MDS}
Let $\mathbf{G}_{\mathrm{IS}}$ be an individually secure encoding matrix as defined in Section \ref{sec:IS}, such that $X^{i}=M^{i}\mathbf{G}_{\mathrm{IS}}$.
Let $M^{i}_{c}$ and $X^{i}_{c}$ denotes any set of $c$ symbols from the columns $M^{i}$ and $X^{i}$, and $M^{i}_{l-c}$ and $X^{i}_{l-c}$ the remaining symbols, respectively. For each $i$-th column, given $M^{i}_{c}$ and $X^{i}_{l-c}$ one can obtain $X^{i}_{c}$ in $\mathcal{O}(l^3)$ operations.
\end{lemma}

\begin{proof}
The individually secure linear code defines $2^{u(\ku-c)}$ cosets. One of them is the code itself. Each coset contains $2^{uc}$ codewords. The vector $M^{i}_{c}$ is used as the index to choose the coset, and $M^{i}_{l-c}$ is used to select the codeword within the coset. Since,
given $M^{i}_{c}$, we known the specific coset selected, from $X^{i}_{l-c}$ we can know the codeword selected in the coset. Note that $X^{i}_{l-c}$ contain exactly $uc$ bits, which is the same number of bits that represent the possible codewords in each coset. Such that using the index of the codeword we obtain $M^{i}_{l-c}$. Now that we have the entire column $M^{i}$, we can use the encoding matrix $\mathbf{G}_{\mathrm{IS}}$ to obtain $X^{i}_{c}$. We can represent this as solving the linear system,
\begin{align*}
    \scriptsize
    \left[ \begin{array}{ccccc}
        M^{i}_1 G_{IS}^{\star}(1,1)& + \cdots + & M^{i}_c G_{IS}^{\star}(c,1) + \underline{M^{i}_{c+1}} G_{IS}^{\star\star}(1,1) & + \cdots + &  \underline{M^{i}_{l}} G_{IS}^{\star\star}(l-c,1) \\
        \vdots &  \cdots  & \vdots &  \cdots  & \vdots\\
        M^{i}_1 G_{IS}^{\star}(1,c)& + \cdots + & M^{i}_c G_{IS}^{\star}(c,c) + \underline{M^{i}_{c+1}} G_{IS}^{\star\star}(1,c) & + \cdots + &  \underline{M^{i}_{l}} G_{IS}^{\star\star}(l-c,c) \\
        M^{i}_1 G_{IS}^{\star}(1,c+1)& + \cdots + & M^{i}_c G_{IS}^{\star}(c,c+1) + \underline{M^{i}_{c+1}} G_{IS}^{\star\star}(1,c+1) & + \cdots + &  \underline{M^{i}_{l}} G_{IS}^{\star\star}(l-c,c+1) \\
        \vdots &  \cdots  & \vdots &  \cdots  & \vdots\\
        M^{i}_1 G_{IS}^{\star}(1,l)& + \cdots + & M^{i}_c G_{IS}^{\star}(c,l) + \underline{M^{i}_{c+1}} G_{IS}^{\star\star}(1,l) & + \cdots + &  \underline{M^{i}_{l}} G_{IS}^{\star\star}(c-l,l) \\
    \end{array}\right]
    =\left[ \begin{array}{c}
        \underline{X^i_1}\\
        \underline{\vdots}\\
        \underline{X^i_c}\\
        X^{i}_{c+1}\\
        \vdots\\
        X^{i}_{l}
    \end{array}\right]
\end{align*}
where the variables marked with underline are the unknowns. By using Gaussian elimination, we can solve the linear system in $\mathcal{O}(l^3)$ operations. Hence, the number of operations needed to obtain $X^{i}_{c}$ given $M^{i}_{c}$ and $X^{i}_{l-c}$ is $\mathcal{O}(l^3)$.
\end{proof}

In addition to being computationally secure, if Eve is a weak eavesdropper, observing no more than $l - w$ of the paths, HUNCC is information-theoretically secure.

\begin{theorem}
Algorithm \ref{alg:hybrid_scheme} is $(l,w)$-individually secure.
\end{theorem}

\begin{proof}
The proof follows directly from the individual security of the code $\textbf{G}_{IS}$ as given in \cite[Section VI]{cohen2018secure}.
\end{proof}

We now show the information rate of HUNCC.

\begin{theorem}\label{theo:rate}
Let $1 \leq c \leq l$ be the number of encrypted paths using the public-key cryptosystem $(\mathrm{Enc}, \mathrm{Dec},p_i,s_i,\kb,\nb)$. Then, the information rate is:
\begin{equation}\label{eq:rate_n}
R = \frac{c \cdot\frac{\kb}{\nb} + (l-c)}{l}
\end{equation}
In particular, $\mathcal{R}\rightarrow 1$ with convergence rate $\mathcal{O}(\frac{1}{l})$.
\end{theorem}
\begin{proof}
The total rate in the multipath network with $l$ paths is obtained directly summing the rate $\kb/\nb$ of the $c$ encrypted paths with the rate of 1 in all the remaining $l-c$ paths.
\end{proof}

\noindent A few important remarks on HUNCC are in order:

First, it is essential to note that the public-key used in the $c$ encrypted paths of HUNCC is just the traditional public-key from the underlying cryptosystem. For instance, the one provided in the McEliece cryptosystem. Hence, this public-key is independent of the $\ku$-messages transmitted over the multipath network and can be supplied to Alice in advance over a public channel. Moreover, unlike in information-theoretic security where a unique secret-key is utilized per message transmitted, as is the case for a one-time pad \cite{shannon1949communication}, the public-key in HUNCC can be used for multiple messages. In the same manner, the generation matrix $\textbf{G}_{IS}$ of the individually secure code is not confidential. Alice and Bob can agree on it over a public channel, i.e. this matrix may be revealed to Eve. Also, this matrix may be used indefinitely for any future transmissions.

Another point is related to the security level of HUNCC. As stated in Theorem \ref{theo:level_security}, if the underlying public-key cryptosystem has security level $b$, HUNCC is individually computationally secure with security level at least $b-\frac{\delta}{2^b}$ where $\delta = \mathcal{O}(l^3)$. We note that this value is extremely close to $b$. Indeed, if $b=128$ and we used Gaussian elimination to solve the system pertaining to Lemma \ref{lem:linear_codes_MDS}, HUNCC would need $l > 10^{38}$ to reduce the security level by $1$.

Finally, although our main focus is on post-quantum cryptosystems, HUNCC can be used with any public-key cryptosystem, including ones which are combined with a symmetric-key cryptosystem \cite{agrawal2012comparative,delfs2007symmetric,PQCRYPTO2015}.

In Section \ref{sec:efficiency}, we present a more detailed analysis of the performance of HUNCC. And in Section \ref{sec:Applications} we provide other features and applications of HUNCC.

\section{Performance Analysis}\label{sec:efficiency}

In this section, we analyze the performance of HUNCC. We focus on three measures, the information rate, the individual computational secrecy, and the information-theoretic individual secrecy. In particular, we compare the trade-off between these quantities as we vary the number of encrypted links, denoted by $c$ in Algorithm \ref{alg:hybrid_scheme}.

We start by considering the communication rate, see Theorem~\ref{theo:rate}. We note that the rate increases as we reduce the number of encrypted links $c$.
On the other hand, the computational security level remains virtually constant as long as at least one path is encrypted, i.e. $c \geq 1$.

Finally, in the case of a weak Eve, the individual security level increases linearly with the number of links, $l-w$, from which the weak Eve does not get information. In other words, while the computational bit-level security remains constant, the uncertainty of an adversary which observes a subset of the links increases.\off{\raf{I think we need to discuss our claims here about the information-theoretic individual secrecy. I do not think it increases with the number of encrypted paths. I think we need to be careful here.}\ale{Agree! we can say this work only for weak Eve that see only $w$ paths, and change $c$ with $w$.} \ale{I fixed the above part and I believe that we need to remove the part below (Let me know what do you think):}
\red{In addition, under the stronger notion of individual secrecy, see \eqref{level_is}, the uncertainty of the adversary can be extended to entire subsets of messages.
It shall be noted than an adversary which decides to omit entirely any of the $c$ encrypted links, and consider it to be an erasure, it will fall under the security guarantee of the $(l,c)$-individual security.}}

\begin{remark}
In Algorithm~\ref{alg:hybrid_scheme}, we assume that each path utilizes the same public-key cryptosystem, with identical public and private keys. If we instead opt to use different cryptosystems, each with independent public and private keys, the total computational security does increase, though marginally. Using $c$ public-key cryptosystems increases the security level from $b$ to $b + \log(c)$. It should be emphasized that this increase on the security level comes at the cost of additional private-keys needing to be shared prior to the communication phase.
\end{remark}

\begin{figure}[t]
    \centering
    \includegraphics[scale=0.43]{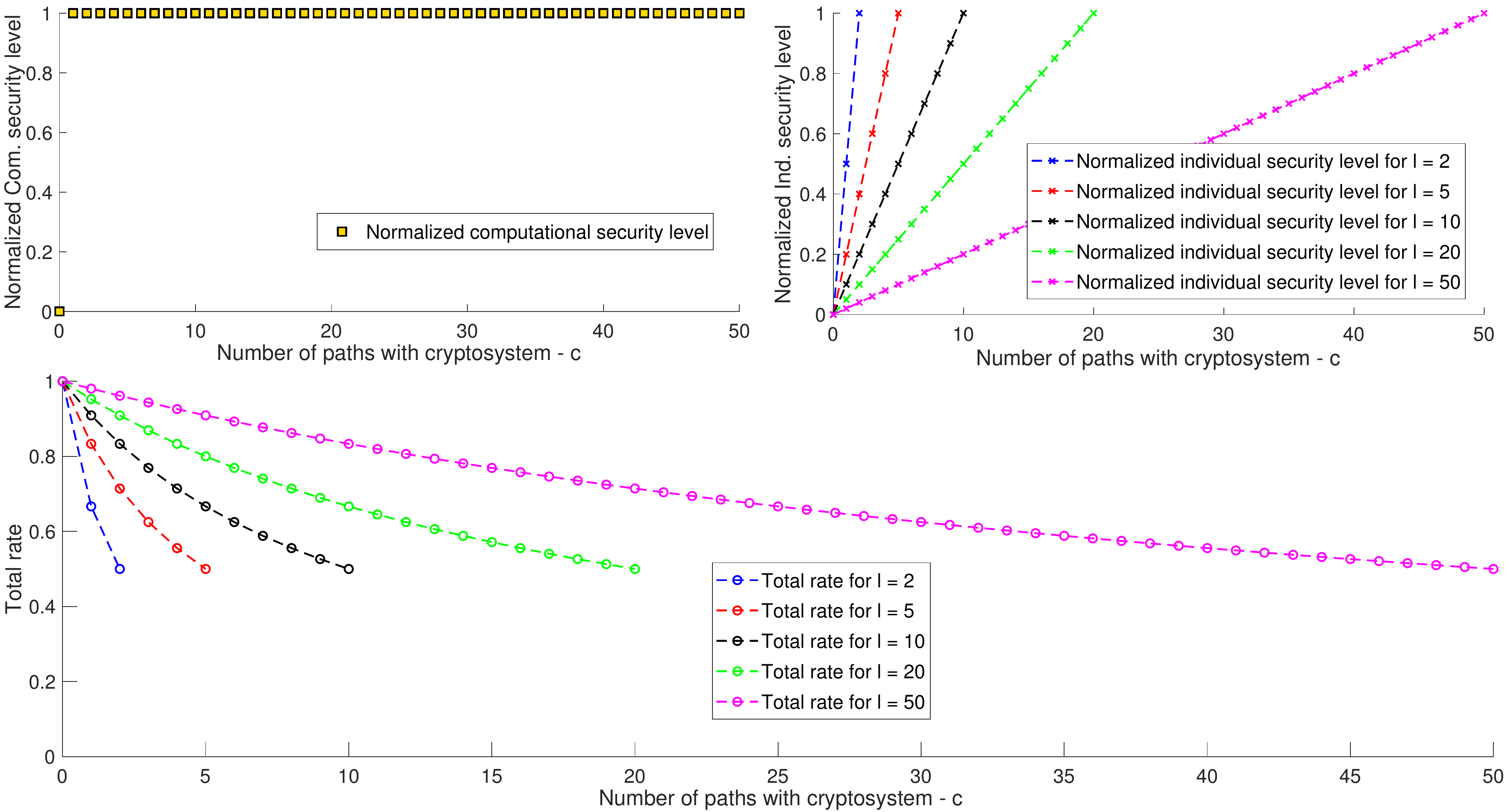}
    \caption{Efficiency and trade-off between security level and communication rate as a function of the number of encrypted links (using a classical McEliece code of rate $\approx 0.5$).}
    \label{fig:trade_off}
    \vspace{-0.2in}
\end{figure}

In Figure \ref{fig:trade_off}, we illustrate each performance parameter versus the number of encrypted links $c$, for HUNCC with the original McEliece cryptosystem using a $[1024,512]$-Goppa code. The information rate over each encrypted path is $\kb/\nb \approx 0.5$. For illustration purposes, we consider normalized measures of security. Namely: 1) the normalized computational security level is the security level divided by the maximum computational security that can be obtained by using the cryptosystem over all links, i.e.
$f^{s}_{crypto} = \frac{\min\{c,1\} \cdot b_{code}}{b_{max}}$.\off{\raf{I think here we need to be careful to include the loss from the theorem.} \ale{Maybe one way is to refer to the remarks before this section when you explained that this loss is negligible} \ale{See the note below.}} 2) the normalized individual security level is the fraction between the number of the links the weak Eve can see divided by the total links in the network, i.e. $f^{s}_{IS} = \frac{l-w}{l}$.
\off{\ale{Here we need to change $c$ with $l-w$} \ale{Done!}}

In the example given in Figure \ref{fig:trade_off}, the results are for the case in which the same McEliece cryptosystem code with the same public-key is used in all the $c$ links.
Hence, the normalized computational security level is zero for $c=0$ and one for any $c=l$.
Note that, via Theorem~\ref{theo:level_security}, the computational security level for any $1 \leq c < l$, is bounded by $b - \delta/2^{\delta}$, where $\delta = O(l^3)$. This difference is negligible and does not appear on the plot, which essentially remains constant and equal to one for $c \geq 1$. \off{\raf{This paragraph seems to be repeating the last one a bit. I think we should merge them.}\ale{I think that this is exactly the explanation that was missing to you in your previous note.}}

\begin{figure}[t]
    \centering
    \includegraphics[scale=0.43]{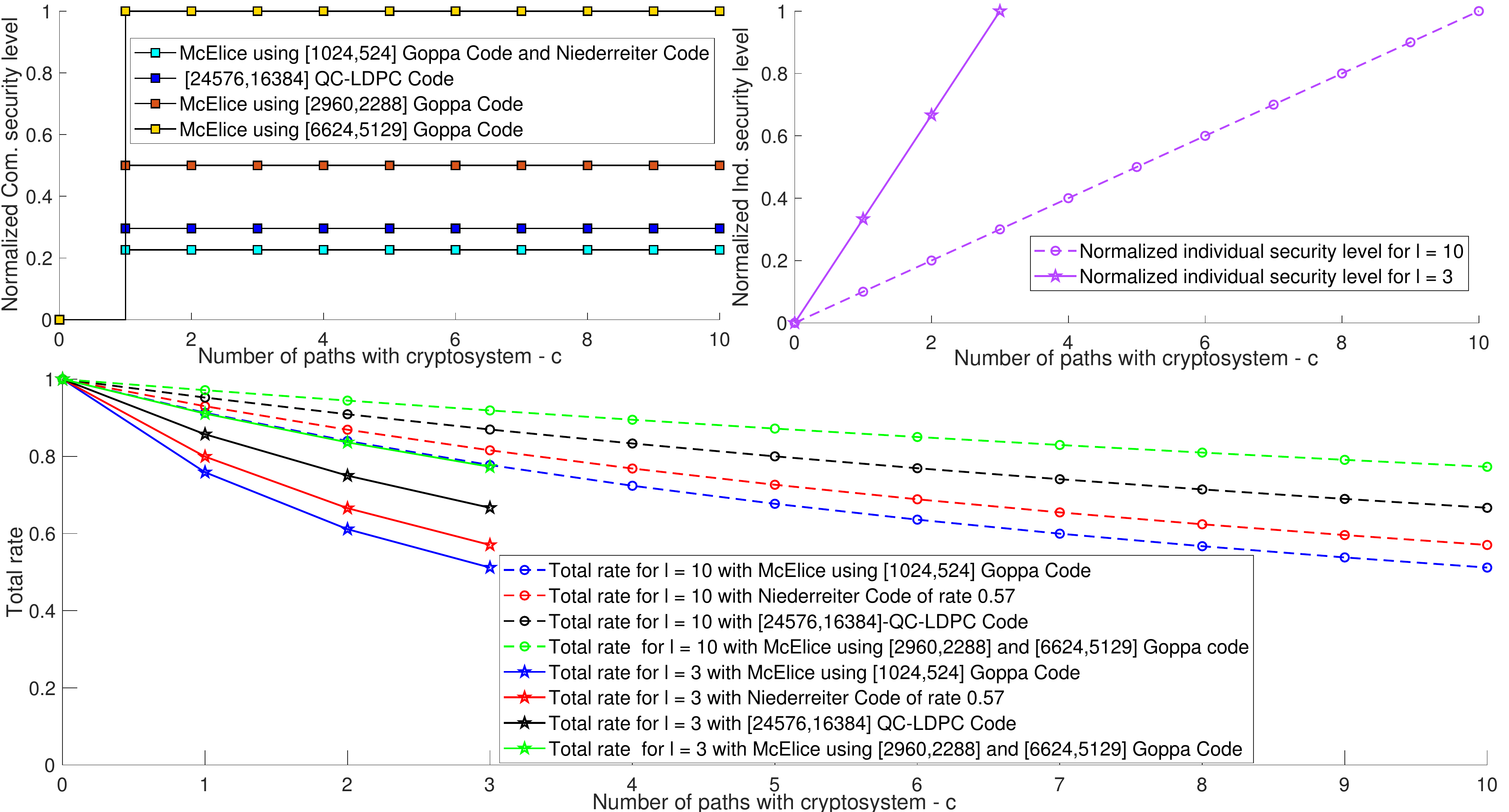}
    \caption{Trade-off between security level and communication rate as a function of the number of encrypted links using different computational security codes.}
    \label{fig:trade_off_codes}
    \vspace{-0.2in}
\end{figure}

The solution proposed in this paper is not restricted to the original McEliece cryptosystem. Any cryptosystem can be utilized. To illustrate this, in Figure \ref{fig:trade_off_codes} we present the efficiency of the proposed solution and the resulting trade-offs for other codes suggested for the McEliece cryptosystem.
In particular, the result presented in Figure \ref{fig:trade_off_codes} are for $l=3$ and $l=10$, and the following codes: 1) McEliece cryptosystem with the original parameters, namely $[1024,524]$-Goppa code which achieves $b=58$-bit computational security. 2)  McEliece cryptosystem with $[2960,2288]$-Goppa code which achieves $b=128$-bit computational security, and $[6624,5129]$-Goppa code for $b=256$-bit computational security as suggested in \cite{bernstein2008attacking}. Both have an information rate of $\kb/\nb \approx 0.777$. Similar parameters for long-term security  in post-quantum cryptography are proposed in \cite{PQCRYPTO2015}, namely McEliece with a $[6960,5413]$-Goppa code. 3) McEliece cryptosystem adopting QC-LDPC codes as presented in \cite{baldi2007cryptanalysis,baldi2009ldpc}. For this family of codes, the following parameters are suggested, $\kb=16384$ and $\nb=24576$, i.e. a code with rate $\kb/\nb = 0.6667$. This code achieves $b=75.8$-bit computational security. 4) Niederreiter in \cite{niederreiter1986knapsack} proposed a Reed-Solomon code with a rate of $0.57$, however the generalized Reed-Solomon codes were broken \cite{sidelnikov1992insecurity}. In \cite{li1994equivalence}, a Niederreiter-type system, which utilized the same Goppa codes used by original McEliece construction, and with the same security level, was proposed. That encryption scheme was also considered in \cite{bernstein2008attacking} and tested under the state-of-the-art attacks of the McEliece cryptosystem, and it thus assumed to achieve $b=58$-bit computational security.

In this comparison, presented in Figure \ref{fig:trade_off_codes}, the maximum computational security level obtained is of $256$-bit using the McEliece cryptosystem with $[6624,5129]$-Goppa codes. Hence, the results presented on the left in Figure \ref{fig:trade_off_codes} for each possible code are normalized with $b_{max}=256$. Recall that while the computational security level is one of the main parameters considered in choosing a cryptosystem, another one is the size of the public-key.
We note also that the computational security level of HUNCC remains essentially constant despite increasing the number of links that use the cryptosystem or by changing the total number of paths in the network.

\section{Other Features and Applications}\label{sec:Applications}
In this section, we present a few features and applications that exemplify the applicability and the performance of the proposed HUNCC scheme for extensive essential applications.

\subsection{Single Path Communication}

While the main network scheme considered in this work is with multi-paths, the secure coding scheme we propose in this work is universal in the sense that it can be applied for any communication network \cite{silva2011universal}. For instance, it can be used in classical point-to-point single-path communication, as well as heterogeneous mesh networks. Figure \ref{fig:Genrral_Scheme_sp} depicts how the proposed secure scheme can be applied in classical point-to-point single-path communication. In this setting, similar to the model presented in Section \ref{sec:model}, there is one source Alice, one legitimate destination Bob, and an eavesdropper Eve. The main difference is that, in the model presented in this section, there is only one path to transmit information between Alice and Bob.
Alice wants to transmit securely over the single-path communication a message $M$ of $\ku$ symbols over a finite field $\mathbb{F}_{q^u}$.
For this, we propose that Alice and Bob utilize HUNCC as described in Section \ref{sec:hybrid}, by essentially simulating parallel virtual links.
More precisely, at the first stage Alice will encode the $\ku$ symbols using the linear individual secure code, $\mathbf{G}_{IS} \in \mathbb{F}_{q}^{l \times l}$.
Then, Alice encrypts $c$ of the symbols before the transmission over the channel using the public-key provided by Bob over the public directory.
The remaining process is as given in Section \ref{sec:hybrid}.
Hence, we obtain the same communication rate and security level as in the multipath network case.
\begin{figure}
    \centering
    \includegraphics[scale=0.5]{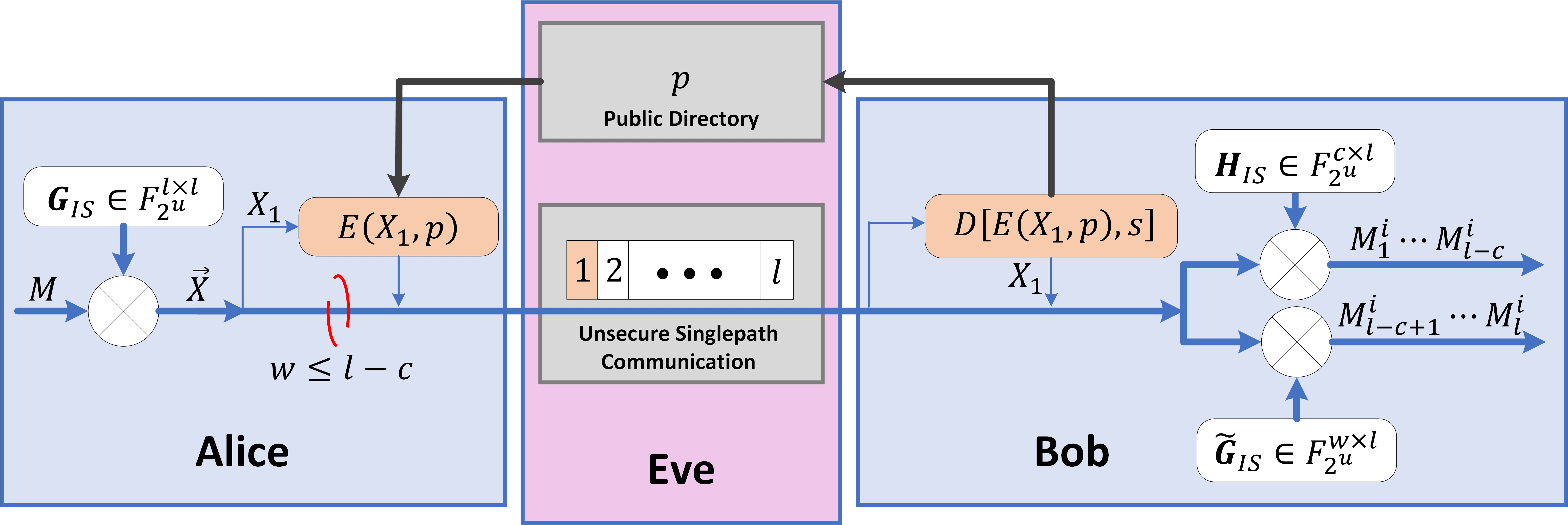}
    \caption{Hybrid post-quantum secure singlepath communication scheme.}
    \label{fig:Genrral_Scheme_sp}
    \vspace{-0.2in}
\end{figure}

\subsection{Myopic Adversaries} \label{sec: aplications myopic}

An additional important scenario considered in the literature is one in which Eve is allowed not only to eavesdrop the information transmitted over the network, but also to corrupt the encrypted packets.
This scenario is considered in the literature under different models of adversaries, for instance, with passive attacks \cite{silva2011universal,zhang2010p,liu2018security}, myopic adversaries  \cite{sarwate2010coding,dey2015sufficiently,dey2019sufficiently,dey2019interplay}, man in the middle attack \cite{meyer2004man,wagoner2011detecting}, byzantine attacks \cite{ho2008byzantine,he2009secure}, etc.
\begin{figure}[t]
    \centering
    \includegraphics[scale=0.47]{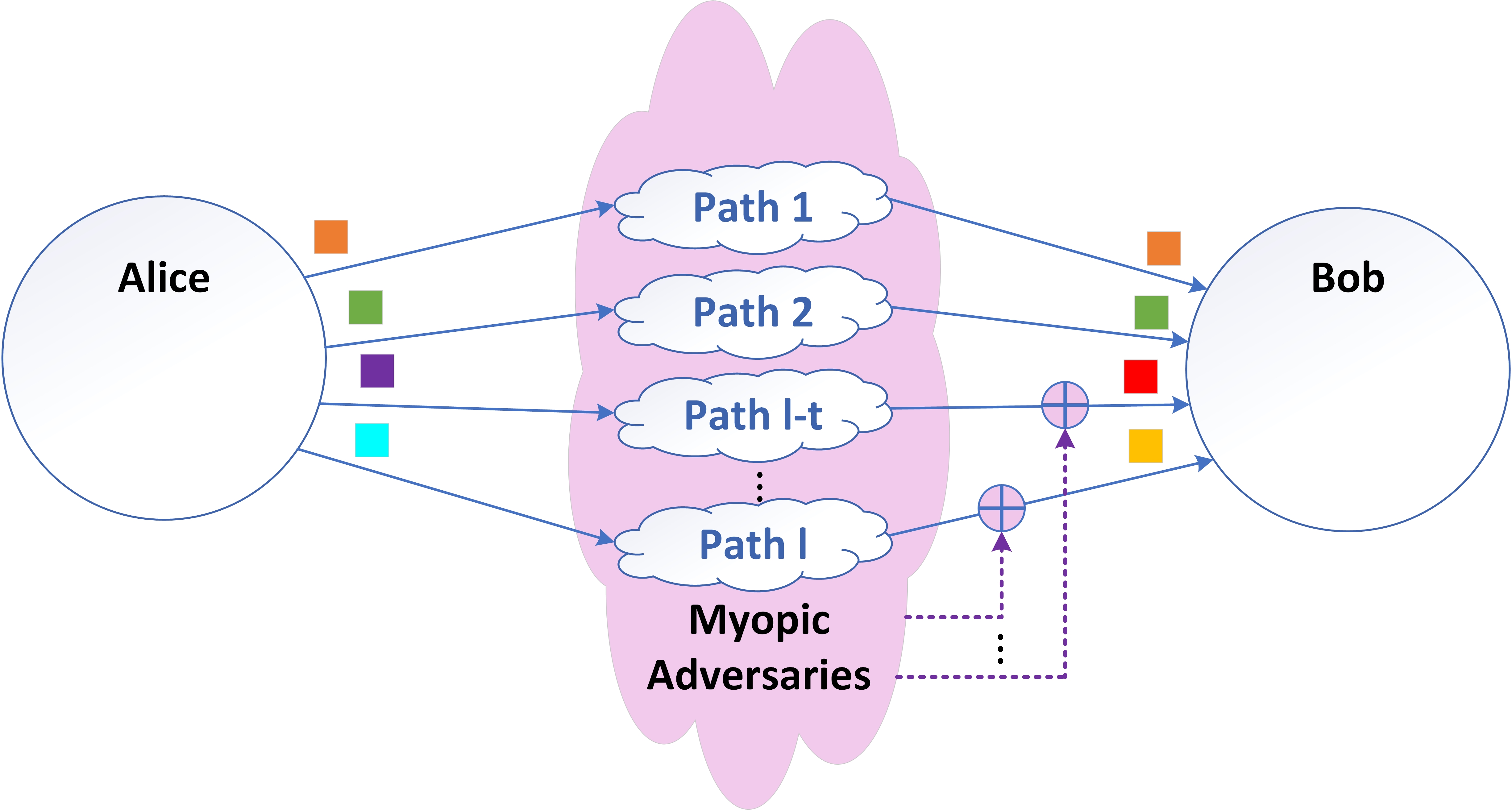}
    \caption{Myopic adversaries model. The eavesdropper is available not only to eavesdrop the information transmitted over the network but also to inject errors into $2t$ packets sent from Alice to Bob that corrupt those encrypted packets.}
    \label{fig:MyopicAdversaries}
    \vspace{-0.2in}
\end{figure}

In the case of a weak eavesdropper, which can obtain only information from $w < l$ subset of the paths in the network; we can augment our linear individual security code to perform error correction, which allows to recover up to $t$ errors which may be injected by Eve.
A possible solution is to generalize the code in the same manner as presented in \cite[Sections IV and VII]{silva2011universal}.
Such extension comes at a cost.
To correct $t$ injected errors, the rate must be decreased by $2t$.
Note that in this setting, the code can support the case in which Eve can corrupt any subset of messages transmitted over the network, whether they are on the paths which are encrypted, or not.
Indeed, the correction property only relies on the decoding of the linear coding scheme, and is independent of the deciphering phase.

In the case of a strong eavesdropper, which can obtain the information from all the paths in the network, we can utilize the same generalized code to correct the $t$ errors that may be injected by Eve. However, in this case, to ensure security, Alice will need to encrypt at least $2t+1$ messages that are transmitted over the different paths in the network.
Hence, we must encrypt the additional $2t$ messages, transmitted to correct the errors, to prevent Eve from obtaining sufficient encoded messages by the linear code, which may provide her with the needed rank to decode the total message.

Future work can consider the case where authentication can be utilized between the encoded messages to reduce the overhead required to correct the injected errors in the above-proposed solution. Note that if Bob is able to identify the corrupted messages, Alice needs to include only one additional symbol per injected error, as opposed to the two messages in the model presented above. Furthermore, the linear code can be generalized as shown in \cite{silva2011universal} to support the scenario where the paths in the network are not error-free \cite{roth1991maximum,silva2011universal}. In the case where the cryptosystem is based on error-correction codes as in the McEliece cryptosystem, instead of adding an error vector at the source, Alice can use the errors of the channel to confuse Eve. The codes in those cases are designed to be able to decode at the legitimate decoder, given those errors. Hence Bob will be able to decode the information.
All of those extensions allow us to increase the effective rate of those solutions.

\subsection{Distributed Storage and Cloud Applications}

The goal of a distributed storage system is to provide reliable access to data which is spread over unreliable storage nodes \cite{dimakis2010network, dimakis2011survey,rashmi2011optimal}. Applications involving data centers are ubiquitous today, Google's GFS \cite{ghemawat2003google}, Amazon's Dynamo \cite{decandia2007dynamo}, Google's BigTable \cite{chang2008bigtable}, Facebook's Apache Hadoop \cite{borthakur2011apache}, Microsoft's WAS \cite{calder2011windows}, LinkedIn's Voldemort \cite{auradkar2012data} and SkyFlok \cite{skyflok}, are just a few examples.

One of the main drawbacks of distributed storage is that, storing the data at more locations can potentially increase the risk in the security and privacy of the data \cite{kadhe2014weaklyNetCod,kadhe2014weakly,paunkoska2016improved}. One way to address this is by reinterpreting the problem as a multipath network. This is done by considering Alice and Bob to be the same individual at different times, and the communication links to be the storage nodes. In this way, the different privacy solutions to the multipath network, including HUNCC, can be readily applied to secure the data in a distributed storage system. Dealing with erasures and errors, of both a probabilistic or adversarial nature, can be done analogous to the way presented in Section \ref{sec: aplications myopic}.

\begin{figure}[t]
    \centering
    \includegraphics[scale=0.5]{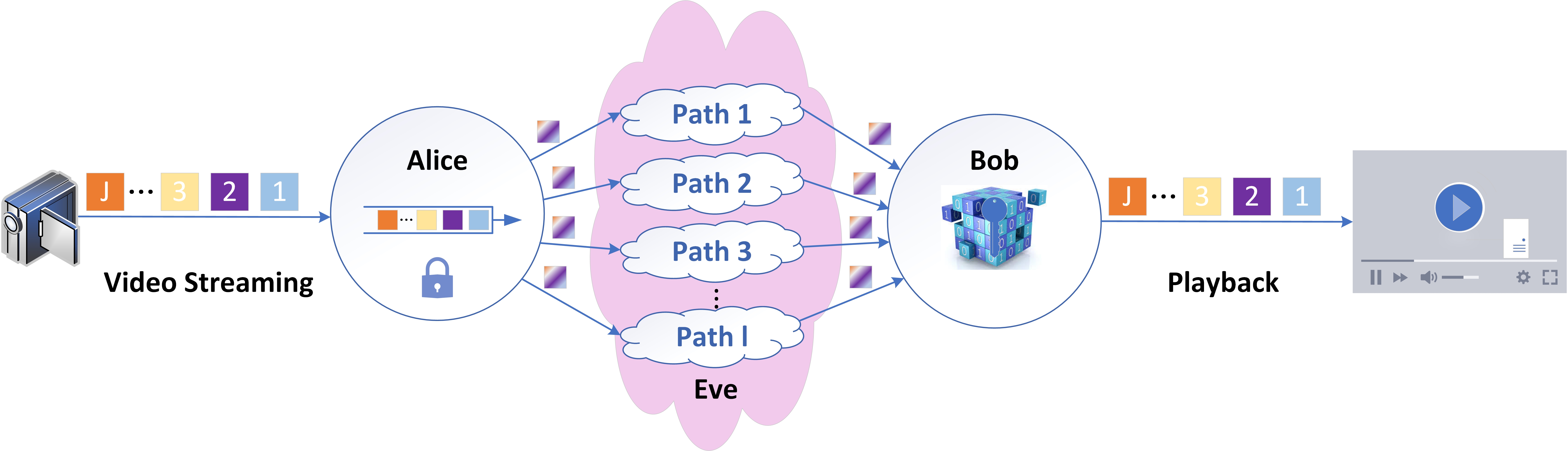}
    \caption{Ultra-reliable low-latency communications for video streaming.}
    \label{fig:Ultra_Reliable}
    \vspace{-0.2in}
\end{figure}

\subsection{Ultra-Reliable Low-Latency Communications}
Recently the application of network coding in streaming communication which demand low delays has been considered \cite{zeng2012joint,KarLei2014,9076631,cohen2019adaptive,MSME19,arno2019discrete}, with applications in audio/video \cite{liveu2017site,liveu2017White}, smart-city \cite{saadat2018multipath,noauthor_snob-5g_nodate}, IoT networks and control applications \cite{andrews2014will,GSMA2018White}, distributed computation \cite{malak2019distribute}, etc.
Figure \ref{fig:Ultra_Reliable} depicts a multi-path low-latency communication for video streaming.
Traditional coding solutions which achieve high throughput are generally not suitable to guarantee low in-order delivery delay which is a requirement in these applications.
This has led to a series of work that propose methods to explore the trade-off between high-rate, and low-delay, see e.g. \cite{cloud2015coded,9076631,cohen2019adaptive,MSME19,arno2019discrete}.


When the communication needs to also be secure, we propose to use HUNCC scheme of Section~\ref{sec:hybrid}, in conjunction with the coding schemes of  say \cite{9076631,cohen2019adaptive,MSME19,arno2019discrete}.
While in those works the number of messages from Alice, that are involved in the linear network encoding process, depend on the desired rate/delay trade-off, in the security application,  this number is further constrained by the security guarantees that are desired.
This might come at the cost of delay, as more messages may need to be mixed in together to provide secrecy.

Future works of interest include the study of new secure codes for streaming communications which allow to optimally trade-off delay, security guarantees, and information rate.

\subsection{RSA Cryptosystem}\label{rsa_example}

Although in this work, we focus on post-quantum cryptography, and specifically on the McEliece coding-scheme in the $c$ encrypted paths, any computationally secure cryptosystem can be used for our solution. Here we briefly discuss an example of how RSA \cite{rivest1978method} can be applied in our network-coding solution, in the context of the example given in Figure \ref{fig:foobar}-(e), i.e. a multipath network with two paths.
As a first stage of encoding at Alice, we assume that the generation matrix,  $\textbf{G} = \left( \begin{smallmatrix} 1 & 1\\  2 & 1 \end{smallmatrix} \right )$, of the individual security code is used. Such that, $X = M\textbf{G} = [M_1 + M_2 , M_1 + 2 M_2]$.
Now, say Alice and Bob agree on using an RSA scheme only over the first path. For $128$-bit security, they settle on using a $3072$ bit key. Using a $328$ bit OAEP padding, the message size can be at most $2744$ bits. Thus, Alice can map $X_1$ into a $2744$ bit vector and encode it using RSA into $E(X_1,p) \in \mathbb{F}_{2}^{3072}$. Alice will then send $\log_2 |E(X_1,p)| = 3072$ bits through channel 1 and $\log_2 |X_2| \leq 2288$ bits through channel 2. Thus, the total communication cost will be around $5360$ bits giving a communication rate slightly greater than $0.85$.

\section{Conclusions}\label{sec: conclusions}
In this work, we suggests a novel secure post-quantum cryptographic scheme that achieves high communication rates without sacrificing the computational security level.
Surprisingly, this solution based on premixing the information can provide the same security level suggested for post-quantum security,  even if only a single link is encrypted.
This is particularly appealing when a classical public-key cryptosystem (e.g., McEliece coding scheme) can only be used in part of the data transmitted or stored.
In addition to this flexibility, we showed that the information rate converges to $1$ as $\mathcal{O}(1/l)$, thus greatly improving upon the current post-quantum cryptosystems.
This scheme demonstrates the potential for hybrid solutions, which combines information-theory security with public-key cryptography, as a way to improve the performance of the security schemes in terms of the information rate and security trade-off.

\bibliographystyle{IEEEtran}
\bibliography{references, Ref1, Ref2}
\end{document}